\documentclass{amsart}[12pts]

\usepackage{amsmath}
\usepackage{graphicx}
\usepackage{comment}

\usepackage{color}
\newtheorem{theorem}{Theorem}[section]
\newtheorem{lemma}[theorem]{Lemma}
\newtheorem{proposition}{Proposition}[section]
\newtheorem{corollary}{Corollary}[section]
\theoremstyle{definition}
\newtheorem{definition}[theorem]{Definition}
\newtheorem{example}[theorem]{Example}

\theoremstyle{remark}
\newtheorem{remark}[theorem]{Remark}
\newtheorem{pf4.7}[theorem]{proof of theorem 4.7}
\numberwithin{equation}{section}

\def\Z{\mathbb{Z}}
\def\LKD{\Lambda_{K_{d}(A)}}
\def\LTD{\Lambda_{T_{d}(A)}}

\def \Kd{K_{d}(A)}
\def \Ld{L_{d}(N)}
\def \DT{DT_{d}(N)}
\def\Tr{\text{Tr}}
\def\d{\mathrm{d}}
\def\dt{\mathrm{d}t}

\DeclareMathOperator*{\LIM}{LIM}
\def\avint{\mathop{\,\rlap{-}\!\!\int}\limits}
\begin{document}

\title{Asymptotic analysis of determinant of discrete Laplacian}

%  Information for first author
\author{Yuhang Hou}
\author{Santosh Kandel}

\begin{abstract}
In this paper, we study the relation between the partition function of the free scalar field theory on hypercubes with boundary conditions and asymptotics of discrete partition functions on a sequence of ``lattices'' which approximate the hypercube as the mesh approaches to zero. More precisely, we show that the logarithm of the zeta regularized determinant of Laplacian on the hypercube with Dirichlet boundary condition appears as the constant term in the asymptotic expansion of the log-determinant of the discrete Laplacian up to an explicitly computable constant. We also investigate similar problems for the massive Laplacian on tori.

\end{abstract}
\maketitle

\section{Introduction}

Lattice field theories have been very successful to study non-perturbative problems in Quantum Field Theories (QFT). The results they produce not only agrees with the experiments, but also give insights on quantum field theories \cite{rothe2012lattice}. For example, Ising models at the critical temperature has been used to probe conformal field theories, in particular, the scaling limit of a Ising model at the critical temperature captures certain aspects of conformal field theories \cite{hongler2013energy, dubedat2009sle}. Despite the notion of scaling limit for discrete Gaussian QFTs is not well understood in general, these studies suggest that the asymptotic behavior of discrete Gaussian QFT associated to triangulations of a compact Riemannian manifold may be used to study the free Bosonic QFT on the manifold as the mesh of the triangulation approaches to zero. In particular, it may be used to construct a free Bosonic QFT on a compact Riemannian manifold as a ``scaling limit" of the discrete Gaussian QFT as the mesh becomes smaller and smaller. This motivates the main goal of this work, which is to analyze the asymptotic behavior of the discrete partition functions and to study whether the asymptotic expansion contains any information about the partition function of the continuum theory. In the discrete case, the space of fields associated to a finite lattice is a finite dimensional vector space. Hence, for a discrete Gaussian QFT on a finite ``lattice'', the partition function can be expressed in terms of the determinant of an operator on a finite dimensional vector space. However, the space of fields for the free Bosonic QFT on a Riemannian manifold
is an infinite dimensional vector space and more importantly, the Laplacian, which is used to define action functional, is unbounded operator. Hence, we need a notion of determinant which generalizes the usual notion of determinant in order to define the partition function of the theory. In
\cite{hawking1977zeta}, Hawking suggests that the notion of zeta regularized determinant can be used to define the partition function of a free Bosonic QFT. If we use the zeta regularized determinant of Laplacian on a compact Riemannian manifold to define partition function of free Bosonic QFT on the manifold, it is reasonable to expect that the asymptotic expansion of (log) determinant of the discrete Laplacian (with respect to the mesh) may contain some information about the (log) zeta regularized determinant of the Laplacian.

The study of asymptotic for determinants of discrete Laplacians has a long history \cite{kirchhoff1847ueber,kasleleyen1963,duplanticr1988exact} and there are several partial results. Kasteleyn \cite{kasleleyen1963} computed a partial asymptotic expansion on a two dimensional torus. Duplantier-David \cite{duplanticr1988exact} derived a partial asymptotic expansion for the log of the determinant of corresponding discrete Laplacian on rectangular domains in $\mathbb{R}^2$ and two dimensional torus and showed that the zeta regularized determinant of the Laplacian (on rectangular domains and on two dimensional tori) appears in the constant term of the asymptotic expansion. Kenyon \cite{kenyon2000asymptotic} derived a partial asymptotic expansion for the determinant of the corresponding discrete Laplacian on rectilinear polygonal domains,
however, he did not calculate the constant term of the expansion. For a general Riemannian manifold, the existence of a constant term in that partial expansion, let alone its identification with the zeta-regularized determinant, remains an open problem. If the spectrums of the discrete Laplacian and Hodge Laplacians are explicitly known, then it might be possible to get finer asymptotic results. This is the case for tori, which has been studied by Kasteleyn \cite{kasleleyen1963} and Duplantier-David \cite{duplanticr1988exact} in dimension two and by
Chinta, Jorgenson and Karlsson \cite{chinta2010zeta} in the general case. It is also shown in \cite{chinta2010zeta} that the constant term in the asymptotic for the determinant of discrete Laplacians is the logarithm of the zeta-determinant generalizing previous results by Duplantier-David \cite{duplanticr1988exact}. A similar problem, the analysis of an asymptotic expansion of log determinant of discrete Laplacians with free boundary condition, is studied by Louis \cite{louis2017asymptotics} based on generalization of the method used by Chinta, Jorgenson and Karlsson \cite{chinta2010zeta}. A different method, which is based on polyhomogeneous expansion of resolvent trace and regularized limit, developed by Vertman \cite{vertman2015regularized} calculates the constant term of the asymptotic for log of the determinant of discrete Laplacian on tori. Furthermore, Vertman's answer agrees with the result of Chinta, Jorgenson and Karlsson \cite{chinta2010zeta} on tori.

In this paper, we first consider the discrete Laplacian on the hypercube with Dirichlet boundary condition and give an asymptotic expansion of the determinant of the discrete Laplacian. We show that the log of the zeta regularized determinant of the continuous Laplacian appears in constant term of the expansion. Furthermore, we modify Vertman's approach to compute the constant term of the asymptotic expansion for the hypercube case with both free and Dirichlet boundary condition. Finally, we investigate a similar problem for massive Laplacian.

\subsection{Summary of the main results}
First, let us first introduce some notation which will appear in the discussion of the main results in this section of this paper. Let $a_1,...a_d>0$, $A:=(a_1,...a_d)$.
$$K_{d}{(A)}:=K_{d}{(a_1,...a_d)}=[0,a_1]\times...\times[0,a_d]$$
denote the $d$-dimensional hypercube. Let $n_i(u), i=1,\dots,d$ be positive integers such that $$\lim_{u\rightarrow\infty}\dfrac{n_i(u)}{u}=a_i$$ and $N$ denote the $d$-tuple $(n_1,\dots, n_d).$ We note that $N$ is a function of $u$ but we will not mention it explicitly unless it is needed.

Let $L_{d}(N)$ denote the $d$-dimensional orthotope (discrete hypercube) which is the product of $d$ path graphs $P_{n_i(u)}$, $i=1,\dots, d$. Also, Define $V_{k}^{d,N}$ for $k=0,1\dots,d$ by 
$$V^{d,N}_k:=\frac{1}{2^{d-k}}\sum_{0< i_1<\cdots<i_k\leq d}\prod^k_{j=1}n_{i_j}$$ 
and 
$$V^d_k:=\frac{1}{2^{d-k}}\sum_{0< i_1<\cdots<i_k\leq d}\prod^k_{j=1}a_{i_j}.$$
We note that $2^{d-k}V^d_k$ is the sum of the volumes of $k$-dimensional faces of $K_d(A)$ \cite{louis2017asymptotics}. 

The discrete Laplacian $\Delta_{L_d(N)}$ on $L_d(N)$ is defined as follows: Let $f$ be a function defined on the discrete hypercube $L_d(N)$, more precisely on the set of vertices, then 
$$
\Delta_{L_d(N)} f(x):=\sum_{y\sim x}(f(x)-f(y)).
$$
If we consider $L_d(N)$ as a graph, then $\Delta_{L_d(N)}$ is the graph Laplacian. In this paper, we are interested in the restriction of $\Delta_{L_d(N)}$ on the functions which vanish on the boundary of ${L_d(N)}.$ We will emphasize such a restriction by calling it the discrete Laplacian with Dirichlet boundary condition. 

The following theorem is one of the main results of the paper:
\begin{theorem}\label{theorem1.1}
Let $\Delta_{L_{d}(N)}$ be the discrete Laplacian on the $d$-dimensional discrete hypercube $L_{d}(N)$ with Dirichlet boundary condition. Let $\Delta_{K_{d}(A)}$ be the Laplacian of the hypercube $K_{d}(A)$ with Dirichlet boundary condition. Then,
\begin{equation}\label{mainteq1}
\begin{aligned}
  \log \det\Delta_{L_{d}(N)}&=\sum^d_{i=1}V^{d,N}_{i}\mathcal{L}^d_{i}(0)-\frac{(-1)^d}{2^d}\log(u^2)+\log {\det}_{\zeta}\Delta_{K_{d}(A)}\\
  &+\frac{(-1)^d}{2^d}\sum^d_{i=1}\log(4i)\dbinom{d}{i}+o(1)
  \end{aligned}
\end{equation}
as $u\rightarrow\infty$, where
$$
\mathcal{L}^d_{i}(0)=-\int ^{\infty}_0\left((-1-e^{-4t})^{d-i} e^{-s^2t}e^{-2it}(I_0(2t))^i-(-2)^{d-i}e^{-t}\right)\frac{\mathrm{d}t}{t}\quad 0<i\leq d,
$$ \text{and} $\log {\det}_{\zeta}\Delta_{K_{d}(A)}$ is the zeta-regularized determinant.
\end{theorem}

\begin{remark}\label{Remark1.1}
The first term of (\ref{mainteq1}) has contributions from the bulk and the boundary, whereas the term in the second line is the contribution from the corner.
\end{remark}

\begin{remark}\label{Remark1.2} In \cite{louis2017asymptotics}, Louis proves a similar result for the discrete Laplacian with the free boundary condition. However, our approach in this paper differs from that of \cite{louis2017asymptotics} (see remark 2.7). 
\end{remark}

Let us consider the special case when $n_i(n)=n$ for all $i=1,2,\dots, d$ where $n\in \mathbb{N}\setminus\{0\}$. In this case, $N(n)$ is the $d$-tuple $(n,\dots, n)$. Let us consider the rescaled discrete Laplacian $\tilde{\Delta}_{L_d(N)}=n^2 \Delta_{L_{d}(N)}$ with the Dirichlet boundary condition. Then, $\log \det\tilde{\Delta}_{L_{d}(N)}$ is a function of $n$. Modifying arguments in \cite{vertman2015regularized}, we show that $\log \det\tilde{\Delta}_{L_{d}(N)}$ has ``regularized limit", denoted by $\LIM$, as $n$ approaches infinity. Moreover, we will also prove the following result.

\begin{theorem}\label{Th1.2}
The logarithmic determinant $\log \det\tilde{\Delta}_{L_{d}(N)}$ has a regularized limit as $n\rightarrow\infty$ and 
\[\log \det\nolimits_{\zeta}\Delta_{K_d((1,1,\dots,1))}=\LIM_{n\rightarrow\infty}\log \det\tilde{\Delta}_{L_{d}(N)} -\frac{1}{2^d}\sum^d_{i=1}\log(4i)(-1)^i\dbinom{d}{i}.\]
\end{theorem}

Let $DT_{d}(N)$ denote the $d$-dimensional discrete torus $\prod_{j=1}^{d}\mathbb{Z}/n_j\mathbb{Z}$ and $T_d(A)$ denote the $d$-dimensional torus $\mathbb{R}^d/\text{diag}(a_1,\dots,a_d)\mathbb{Z}^d$. Let $\Delta_{DT_{d}(N)}^{m}=\Delta_{DT_{d}(N)}+m^2$ and $\Delta_{T_d(A)}^m=\Delta_{T_{d}(A)}+m^2$ be the massive Laplacians on $DT_{d}(N)$ and $T_{d}(A)$ respectively, where $\Delta_{DT_{d}(N)}$ and $\Delta_{T_{d}}(A)$ are Laplacians on $DT_{d}(N)$ and $T_{d}(A)$ respectively. We also prove the following result.

\begin{theorem}\label{Th1.3}
Let $ \log \det{\Delta^{\tilde{m}}_{DT_{d}(N)}}$ be the log-determinant of the massive Laplacian on the $d$-dimensional discrete torus and $\log \det(\Delta_{T_d(A)}+m^2)$ be the log of the zeta regularized determinant of the massive Laplacian the torus. Then,
\begin{eqnarray*}
  \log \det{\Delta^{\tilde{m}}_{DT_{d}(N)}}=V^{d,N}_d\mathcal{L}_{\tilde{m}}(0)+\mathcal{H}_{N(u)}(0),
\end{eqnarray*}
where 
$$
\begin{aligned}
&\mathcal{H}_{N(u)}(0)\\
&=\log \det(\Delta_{T_d(A)}+m^2)\\
&\ \ \ +\left\{
\begin{aligned}
&\frac{V^d_d}{(4\pi/m^2)^{d/2}}\Gamma(-d/2)~&d~odd\\
&(-1)^{d/2}
\frac{2/d+2/(d-2)\cdots+1-\log(m^2)}{(d/2)!}\frac{V^d_d}{(4\pi)^{d/2}}m^d~&d~even\\
\end{aligned}\right. \\ 
&+o(1) 
\end{aligned}
$$
as $u\rightarrow\infty$,
and 
$$\mathcal{L}_{\tilde{m}}(0)=-\int^{\infty}_0(e^{-\tilde{m}^2t}e^{-2dt}I_0(2t)^d-e^{-t})\frac{\mathrm{d}t}{t}.$$

\end{theorem}

Let us point out a key difference between Theorem \ref{Th1.1} and Theorem \ref{Th1.3} due to the presence of the mass term. Unlike the situation in Theorem \ref{Th1.1}, in Theorem \ref{Th1.3} $\mathcal{L}_{\tilde{m}}(0)$ is a function of $u$ because of the presence of $\tilde{m}$ which is the ``mass term". The next result gives some information about $\mathcal{L}_{\tilde{m}}(0)$. 
\begin{theorem} Let $\tilde{m}=m/u$ and $d\geq 2$. Then,
$$
\begin{aligned}
\mathcal{L}_{\tilde{m}}(0)=&-\int^{\infty}_0(e^{-2dt}I_0(2t)^d-e^{-t})\frac{\mathrm{d}t}{t}\\
&-\tilde{m}\int_0^{\infty}\dfrac{\partial f}{\partial\tilde{m}}(0,t)\,\mathrm{d}t\hskip1mm-,\hskip1mm\dots\hskip1mm,\\
&-\tilde{m}^{d-1}\int_0^{\infty}\dfrac{\partial^{d-1} f}{\partial\tilde{m}^{d-1}}(0,t)\,\mathrm{d}t+o(\tilde{m}^{d-1}).
\end{aligned}
$$
where $f(t, \tilde{m})=(e^{-\tilde{m}^2t}e^{-2dt}I_0(2t)^d-e^{-t})/{t}$. Moreover, the terms with odd order derivative approach to zero as $\tilde{m}\to 0$.
\end{theorem}

When $d=2$, there is a very special relationship between determinants of the massive discrete Laplacian on the torus and on the hypercube, which is the following.

\begin{theorem}
Given a discrete torus of size $N=(2n_1,2n_2)$ and a discrete hypercube of size $N'=(n_1,n_2)$, we have for the determinant of the discrete massive Laplacian with Dirichlet boundary condition,
$$\begin{aligned}
&\frac{\det\Delta_{DT_2(N)}^m}{(\det\Delta_{L_2(N')}^m)^4}=(8+m^2)(4+m^2)^2\times\\
&\prod_{m_1=1}^{m_1=n_1-1}\left(6+m^2-2\cos\left(\frac{m_1\pi}{n_1}\right)\right)^2\prod_{m_2=1}^{m_2=n_2-1}\left(6+m^2-2\cos\left(\frac{m_2\pi}{n_2}\right)\right)^2\times\\
&\prod_{m_1=1}^{m_1=n_1-1}\left(2+m^2-2\cos\left(\frac{m_1\pi}{n_1}\right)\right)^2\prod_{m_2=1}^{m_2=n_2-1}\left(2+m^2-2\cos\left(\frac{m_2\pi}{n_2}\right)\right)^2
\end{aligned}
$$
\end{theorem}

\section{Determinants of Laplacian on hypercube and tori}
In this section, we will prove Theorem \ref{theorem1.1}. In \ref{section 2.1}, we will study the heat kernel of the Laplacian on the hypercube with Dirichlet boundary condition, and then use it to calculate the log of zeta regularized determinant. In \ref{section2.2}, we will establish a relation between the heat kernel of discrete Laplacian on a discrete torus and the heat of discrete Laplacian on a discrete hypercube. In \ref{subsection2.2.1}, we will analyze the behavior of the heat kernel and rewrite the log of determinant of the discrete Laplacian as a sum of $V^{d,N}_i\mathcal{L}^d_i$ and $\mathcal{H}_N$. In the next section, we will show how these terms contribute to the asymptotic expansion, in particular, we will show that $\mathcal{H}_N$ can be expressed as a sum of the log of the zeta regularized determinant of the Laplacian and an explicitly computable constant. 

\subsection{\textbf{Theta functions of Laplacian on $K_d(A)$ and $T_d(A)$}}\label{section 2.1}
Recall that the theta function of an operator is defined as the the trace of the heat kernel the operator and the trace of the heat kernel can be obtained from the spectrum of the operator in nice situations \cite{chinta2010zeta}. 

The spectrum $\LKD$ of the Laplacian on $K_d(A)$ with Dirichlet boundary condition is well known \cite{louis2017asymptotics}: $\LKD=\{{q_1^2\pi^2}/{a_1^2}+\cdots+{q_d^2\pi^2}/{a_d^2}|(q_1,\dots,q_d)\in(\mathbb{N}\setminus\{0\})^{d}\}$.
Let $\theta_{K_d(A)}(t)$ denote the theta function of the Laplacian on $K_d(A)$ with Dirichlet boundary condition, then
$$\theta_{K_d(A)}(t)=\sum_{\lambda\in\LKD}e^{-\lambda t}$$

The spectrum $\LTD$ of the Laplacian on $T_d(A)$ is also well known \cite{chinta2010zeta}: $\LTD=\{{(2\pi q_1)^2}/{a_1^2}+\cdots+{(2\pi q_d)^2}/{a_d^2}|(q_1,\dots,q_d)\in\mathbb{Z}^d\}.$ Hence, the theta function of the Laplacian on the torus $T_d(A)$ is given by
$$\theta_{T_d(A)}(t)=\sum_{\lambda\in\LTD}e^{-\lambda t}.$$

From the explicit description of $\LKD$ and $\LTD$, we can see that these two sets are related. We will use this relationship to relate their corresponding theta functions which is the content of the following proposition.

\begin{proposition}\label{proposition2.1} Let $A=(a_1,\dots, a_d)$, then
\begin{equation}\label{eq2.1}
\theta_{K_d(A)}(t)=\frac{1}{2^d}\sum^d_{m=0}(-1)^{d-m}\sum_{0<i_1 <\cdots<i_m\leq d}\theta_{T_m(2a_{i_1}, \dots,2a_{i_m} )}(t)
\end{equation}
where we use the convention that $\theta_{K_{0}}(t)=\theta_{T_{0}}(t)=1$. %and we use $2A_m=[0,a_{i_1}]\times\dots\times[0,a_{i_m}]$.
\begin{proof}
We will use induction on $d$ to prove the proposition. For $d=0$, (\ref{eq2.1}) holds trivially.

Let us rewrite $\theta_{K_{d}(A)}(t)$ as
 $$\theta_{K_{d}(A)}(t)=\sum_{n_i\in \mathbb{N}\setminus\{0\}}e^{-\sum_{i=1}^dn_i^2(2\pi)^2/(2a_i)^2t}.$$
When $d=1$, 
$$\sum_{ n_i\in \mathbb{N}\setminus\{0\}}e^{-n_i^2(2\pi)^2/(2a_i)^2t}=\frac{1}{2}\left(\sum_{n_i\in \mathbb{Z}}e^{-n_i^2(2\pi)^2/(2a_i)^2t}-1\right).$$
Hence, (\ref{eq2.1}) holds in this case as well. Now, assume that (\ref{eq2.1}) holds for all for $k\leq d-1$. Note that
 \begin{eqnarray*}
 \begin{array}{lll}\theta_{K_{d}(A)}(t)=\left(\sum_{n_d\in \mathbb{N}\setminus\{0\}}e^{-n_d^2(2\pi)^2/(2a_d)^2t}\right)\theta_{K_{d-1}(a_1,...,a_{d-1})}(t).
 \end{array}
 \end{eqnarray*}
Now using the induction hypothesis, we see that 
$$
\begin{aligned}
  \theta_{K_d(a_1,...,a_d)}(t)&=\frac{1}{2}\left(\sum_{n_d\in \mathbb{Z}}e^{-n_d^2(2\pi)^2/(2a_d)^2t}\theta_{K_{d-1}(a_1,...,a_{d-1})}(t)-\theta_{K_{d-1}(a_1,...,a_{d-1})}(t)\right)\\
  &=\frac{1}{2}\left(\theta_{T_1(a_{d})}(t)\theta_{K_{d-1}(a_1,...,a_{d-1})}(t)-\theta_{K_{d-1}(a_1,...,a_{d-1})}(t)\right)\\
  &=\frac{1}{2^d}\left(\sum^d_{p=1}(-1)^{d-p}\sum_{0< i_1 <\cdots<i_p=d}\theta_{T_{p}(2a_{i_1},...,2a_{i_p})}(t)\right.\\
  &-\left.\sum^{d-1}_{p=0}(-1)^{d-p}\sum_{0< i_1 <\cdots<i_p\leq d-1}\theta_{T_{p}(2a_{i_1},...,2a_{i_p})}(t)
\right)\\
  &=\frac{1}{2^d}\sum^d_{p=0}(-1)^{d-p}\sum_{0< i_1 <\cdots<i_p\leq d}\theta_{T_{p}(2a_{i_1},...,2a_{i_p})}(t),
\end{aligned}
$$ and this completes the proof.
\end{proof}
\end{proposition}

The relation (\ref{eq2.1}) is very useful to study the asymptotic behavior of $\theta_{K_{d}(a_1,...,a_d)}(t)$, more precisely, the asymptotic behavior can be easily analyzed using the asymptotic behavior of the $\theta_{T(2a_{i_1},...,2a_{i_p})}(t)$ which is already studied in detail in \cite{chinta2010zeta}. We recall from \cite{chinta2010zeta} that 
\begin{eqnarray}\label{eq2.2}
\begin{array}{lll}
\hskip5mm \theta_{T(2a_{i_1},...,2a_{i_p})}(t)=
 \left\{
  \begin{array}{cc}
    \prod_{q=1}^p(2a_{i_q})(4\pi t)^{-p/2}+\mathcal{O}(e^{-c/t}) & t\rightarrow0 \\
    1+\mathcal{O}(e^{-ct}) &t\rightarrow \infty\\
   \end{array}
  \right.
  \end{array}
\end{eqnarray}  

As a corollary of Proposition \ref{proposition2.1} and \ref{eq2.2}, we get the following.

\begin{corollary}\label{corollary2.1} The following holds:
\begin{eqnarray}
\theta_{K_{d}(a_1,...,a_d)}(t)=
 \left\{
  \begin{array}{cc}
   \sum^d_{i=0}(-1)^{d-i}V^d_i(4\pi t)^{-i/2}+\mathcal{O}(e^{-c/t}) & t\rightarrow0 \\
    \mathcal{O}(e^{-ct}) & t\rightarrow \infty\\
  \end{array}
 \right.
\end{eqnarray}
\end{corollary}

\subsubsection{\textbf{Zeta function of Laplacian with Dirichlet boundary condition on $K_d(A)$:}}
We recall that the zeta function is the Mellin transform of the theta function \cite{chinta2010zeta}:
$$\zeta_{\Kd}(s)= \frac{1}{\Gamma(s)}\int^{\infty}_{0}\theta_{K_{d}(A)}(t)t^s\frac{\d t}{t} $$
Note that the integral is well defined whenever $Re(s)>n/2$. To analyze $\zeta_{\Kd}(s)$, it is convenient to write it as follows:
$$
\begin{aligned}
 \zeta_{\Kd}(s) = &\frac{1}{\Gamma(s)}\int^{1}_{0}(\theta_{\Kd}(t)-f(t))t^s\frac{\d t}{t}\\
   &+\frac{1}{\Gamma(s)}\int^{1}_{0}f(t)t^s\frac{\d t}{t}\\
   &+\frac{1}{\Gamma(s)}\int^{\infty}_{1}\theta_{\Kd}(t)t^s\frac{\d t}{t}
\end{aligned}
$$ where
$f(t) =\sum^d_{i=0}(-1)^{d-i}V^d_i\times(4\pi t)^{-i/2}.$

Using corollary \ref{corollary2.1}, it is clear that the integral has a meromorphic continuation in $\mathbb{C}$ and more importantly, it is holomorphic at $s=0$. Moreover, integrating the second term and then taking the derivative $s=0$ we get:
\begin{equation}\label{cont_inter1}
\begin{aligned}
 \zeta'_{\Kd}(0)= &\int^{1}_{0}(\theta_{\Kd}(t)-f(t))\frac{\d t}{t}\\
   &+\int^{\infty}_{1}(\theta_{\Kd}(t))\frac{\d t}{t}\\
     &-\frac{(-1)^d}{2^d}\Gamma'(1)- \sum^n_{i=1}(-1)^{d-i}\frac{2}{i}V^d_{i}\times(4\pi)^{-i/2}.\\
\end{aligned}
\end{equation}
Later we will see that the relation \ref{cont_inter1} will play an important role in the proof of Theorem \ref{Th1.1}.

As a direct consequence of Proposition \ref{proposition2.1}, we get the following relation between zeta functions.

\begin{corollary} Let $\zeta_{T_{d}(A)}$ be the zeta function of the Laplacian on $T_{d}(A)$. Then,
\begin{eqnarray}
\begin{array}{lll}
\zeta_{\Kd}=\sum^d_{p=0}(-1)^{d-p}\sum_{0<i_1<\cdots<i_p\leq d}\frac{1}{2^d}\zeta_{T_{p}(2a_{i_1},...,2a_{i_p})}.
\end{array}
\end{eqnarray}
\end{corollary}

\subsection{Theta function of discrete Laplacian on discrete hypercube and discrete tori}\label{section2.2}
As in the continuum case, the modified theta function of the discrete Laplacian on $L_d(N)$ with Dirichlet boundary condition is given by
$$\Theta_{\Ld}(t)=\sum_{\lambda\in\Lambda_{\Ld}}e^{-\lambda t},$$
where $\Lambda_{\Ld}$ is the spectrum of the discrete Laplacian with Dirichlet boundary condition.
It is well known that $\Lambda_{\Ld}=\{2d-2\cos(\pi q_1/n_1)-\cdots-2\cos(\pi q_d/n_d)|0< q_j<n_j,\quad q_j\in \Z\}$ \cite{duplanticr1988exact}.

Let $N=(n_1,\dots, n_d)$ and let $DT_{d}(N)$ denote the discrete torus $\prod_{j=1}^d\Z/n_j\Z.$ Then the theta function for $DT_{d}(N)$ is given by 
$$\Theta_{\DT}(t)=\sum_{\lambda\in\Lambda_{\DT}}e^{-\lambda t}$$
where $\Lambda_{\DT}=\{2d-2\cos(2\pi q_1/n_1)\cdots-2\cos(2\pi q_d/n_d)|0\leq q_j<n_j\}$ is the spectrum of the discrete Laplacian on $\DT$ \cite{chinta2010zeta}.

As discussed in Section \ref{section 2.1}, it is possible to express $\Theta_{\Ld}(t)$ in terms of $\Theta_{DT_{p}(n_{i_1},\dots, n_{i_p}))}(t)$ as shown in the following proposition.

\begin{proposition}\label{prop22.2}
We have the following relation:
\begin{equation} \label{eq2.5}
\Theta_{\Ld}(t)=\frac{1}{2^d}\sum^d_{p=0}\sum_{0<i_1 <\cdots<i_p\leq d}(-1-e^{-4t})^{d-p}\Theta_{DT_{p}(2n_{i_1},...,2n_{i_p})}(t).
\end{equation}
Here the convention is that $\Theta_{DT_{0}}(t)=\Theta_{L_{0}}(t)=1$.
\end{proposition}
\begin{proof}
%First, we observe that we can rewrite $\Lambda_{\Ld}$ as
%\[\Lambda_{\Ld}=\{2d-2\cos(\pi 2q_1/2n_1)...-2\cos(\pi2 q_d/2n_d)|0< q_j<n_j\}.\]

We will use induction on $d$ as we did in the Proposition \ref{proposition2.1}. When $d=0$, (\ref{eq2.5}) is satisfied trivially. Using the property of cosine function 
\[\cos(2\pi(m+n)/2n)=\cos(2\pi(n-m)/2n),\]
we can rewrite $\Theta_{L_1}(n_i)(t)$ as
\begin{equation}
\Theta_{L_1}(n_i)(t)=\frac{1}{2}\left(\Theta_{DT_{1}}(2n_i)(t)-(1+e^{-4t})\right)\label{interval1}
\end{equation} which is exactly (\ref{eq2.5}) for when $d=1$. Now, we assume (\ref{eq2.5}) holds for all $k\leq d-1$. When k=d, we have 
\begin{eqnarray*}
\begin{aligned}
&\Theta_{L_d(n_1,...,n_d)}(t)\\
&=\Theta_{L_{1}(n_d)}(t)\Theta_{L_{d-1}(n_1,\dots,n_{d-1})}(t)\\
&=\frac{\Theta_{L_{1}(n_d)}(t)}{2^{d-1}}\sum^{d-1}_{p=0}\sum_{0< i_1 <\cdots<i_p\leq d-1}(-1-e^{-4t})^{d-p-1}\Theta_{DT_{p}(2n_{i_1},...,2n_{i_p})}(t)\\
&=\frac{1}{2^{d}}\left(\sum^{d}_{p=1}\sum_{0< i_1 <\cdots<i_p= d}(-1-e^{-4t})^{d-p}\Theta_{DT_{p}(2n_{i_1},...,2n_{i_p})}(t)\right.\\
&\left.+\sum^{d-1}_{p=0}\sum_{0< i_1 <\cdots<i_p\leq d-1}(-1-e^{-4t})^{d-p-1}\Theta_{DT_{p}(2n_{i_1},...,2n_{i_p})}(t)\right)\\
&=\frac{1}{2^d}\sum^d_{p=0}\sum_{0< i_1 <\cdots<i_p\leq d}(-1-e^{-4t})^{d-p}\Theta_{DT_{p}(2n_{i_1},...,2n_{i_p})}(t)
\end{aligned}
\end{eqnarray*}
which completes inductive step in the proof of the proposition.
\end{proof}

\begin{remark} Using the induction as above, one can give a different proof of corollary 2.3 of \cite{louis2017asymptotics} where the free boundary condition is considered. 
\end{remark}

\subsubsection{Analysis of determinants of discrete Laplacians}\label{subsection2.2.1}

Here, we will study asymptotic behavior of determinant of the discrete Laplacian on $L_{d}(N)$. Our discussion here uses a lot of results from section 3 and section 4 of \cite{chinta2010zeta} and it is inspired by section 3 of the same paper. More precisely, we modify a lot of results from section 3 of \cite{chinta2010zeta} to address our need, in particular, to accommodate the boundary conditions in the consideration.

The following function $g(t)$ defined by
\begin{eqnarray} \label{g(t)}
\begin{aligned}
g(t)=\sum_{p=1}^dV^{d, N}_p(-1-e^{-4t})^{d-p}(e^{-2t}I_0(2t))^p
\end{aligned}
\end{eqnarray}
plays an important role in this section. In the next lemma, we analyze how $g(t)$ interacts with $\Theta_{L_{d}(N)}(t)$ as $t\to0$ and $t\to\infty.$ 

\begin{lemma} \label{analysis of g}As $t\rightarrow0$,
$$\Theta_{L_{d}(N)}(t)-g(t)-(-1)^de^{-t}=\mathcal{O}(t).$$ 

Furthermore, as $t\to\infty$, we have
and 
$$g(t)+(-1)^de^{-t}=\mathcal{O}(t^{-{1/2}}).$$

\end{lemma}
\begin{proof}
Recall that 
\[
\Theta_{\Ld}(t)=\frac{1}{2^d}\sum^d_{p=0}\sum_{0<i_1 <\cdots<i_p\leq d}(-1-e^{-4t})^{d-p}\Theta_{DT_{p}(2n_{i_1},...,2n_{i_p})}(t),
\]
and by our convention $\Theta_{DT_{0}}(t)=\Theta_{L_{0}}(t)=1$.

We recall from \cite{chinta2010zeta} that
\begin{equation} \label{eqn2.9}
\theta_{DT_d(N)}-V^{d,N}_de^{-2dt}I_0(2t)^d=\mathcal{O}(t) 
\end{equation} as $t\rightarrow 0$.
Moreover, as $t\rightarrow \infty$,
\begin{equation}\label{eqn2.10}
e^{-2pt}I_0(2t)^p=\mathcal{O}(t^{-p/2}).
\end{equation}
From (\ref{eqn2.9}) and using the definition of $V^{d, N}_p$, we have 
$$\Theta_{L_{d}(N)}(t)-g(t)-\frac{1}{2^d}(-1-e^{-4t})^d=\mathcal{O}(t)$$ as $t\rightarrow 0$. We also have 
$$\frac{1}{2^d}(-1-e^{-4t})^d-(-1)^de^{-t}=\mathcal{O}(t)$$ as $t\rightarrow 0.$
Combining these proves us the first statement of the lemma.

Using (\ref{eqn2.10}), we see that as $t\rightarrow\infty$, we have
\begin{equation}
(-1-e^{-4t})^{d-p}(e^{-2t}I_0(2t))^p=\mathcal{O}(t^{-p/2}).
\end{equation}
This means that the behavior of $g(t)$ as $t\rightarrow\infty$ is governed by the term $(-1-e^{-4t})^{d-1}(e^{-2t}I_0(2t))$
which is of $\mathcal{O}(t^{-1/2})$. On the other hand $(-1)^de^{-t}=o(1)$ as $t\to\infty$. This concludes the proof of the second statement of the lemma.
\end{proof}

In the next few lemmas, we present auxiliary results which will be used later.
\begin{lemma}\label{3.1}
For all $s\in\mathbb{C}$ with $Re(s^2)>0$, we have
$$\sum_{\lambda\in\Lambda_{\Ld}}\frac{2s}{s^2+\lambda}=2s\int^{\infty}_0e^{-s^2t}g(t)dt+2s\int^{\infty}_0e^{-s^2t}[\Theta_{L_{d}(N)}(t)-g(t)]\d t.$$
%where $\lambda$ is the eigenvalue.
\begin{proof}
By definition of $\Theta_{L_{d}(N)}(t)$, we have
$$\sum_{\lambda\in \Lambda_{\Ld}}e^{-\lambda t}=g(t)+(\Theta_{L_{d}(N)}(t)-g(t))$$
Multiplying both sides by $e^{-s^2t}$ and integrating with respect to $t$, we get
\begin{eqnarray*}
\sum_{\lambda}\frac{1}{s^2+\lambda}=\int^{\infty}_0e^{-s^2t}g(t)dt+\int^{\infty}_0e^{-s^2t}[\Theta_{L_{d}(N)}(t)-g(t)]dt.
\end{eqnarray*}
Now, multiplying both sides by $2s$ yields the desired results.
\end{proof}

\end{lemma}

\begin{lemma}\label{3.2}
Let $f(s)$ be given by $$f(s)= \sum_{\lambda\in \Lambda_{\Ld}}\log(\lambda+s^2),$$
then $f(s)$ is uniquely determined by the differential equation
\begin{eqnarray}\label{ODEf}\partial_sf(s)=\sum_{\lambda\in \Lambda_{\Ld}}\frac{2s}{s^2+\lambda}
\end{eqnarray}
and the asymptotic behavior 
$$f(s)=\log(s^2)\prod^d_{i=1}(n_i-1)+o(1)$$
as $s\rightarrow \infty$.
\end{lemma}

\begin{proof} Obviously, $f(s)$ solves the first order ODE (\ref{ODEf}) and it is easy to see $$f(s)=\log(s^2)\prod^d_{i=1}(n_i-1)+o(1)$$
as $s\rightarrow \infty$. Since the equation (\ref{ODEf}) is a first order ODE, $f$ must be unique. 
\end{proof}

\begin{proposition}\label{3.3}
Let $\mathcal{L}^{d}_{i}(s)$ be given by
$$
\mathcal{L}^d_{i}(s)=
-\int ^{\infty}_0\left((-1-e^{-4t})^{d-i} e^{-s^2t}e^{-2it}(I_0(2t))^i-(-2)^{d-i}e^{-t}\right)\frac{\dt}{t}\quad 0<i\leq d,
$$
then $\mathcal{L}^d_{i}$ is the unique function which solves the differential equation
\begin{eqnarray}\label{ODEL}\partial_s\mathcal{L}^d_{i}(s)=2s\int^{\infty}_{0}(1-e^{-4t})^{d-i} e^{-s^2t}e^{-2it}(I_0(2t))^i\,\dt
\end{eqnarray}
and has the asymptotic behavior
$$
 \mathcal{L}^d_{i}(s)=
(-2)^{d-i}\log(s^2)+o(1)\hskip5mm\text{}~~0<i\leq d\\
$$ as $s\to\infty.$

\begin{proof}
Following the argument similar to the proof of Theorem \ref{T3.1}, we observe that $\mathcal{L}^i_d$ is differentiable. It is easy to verify that $\mathcal{L}^d_{i}(s)$ satisfies the first order ODE (\ref{ODEL}). To study the asymptotic behavior of $\mathcal{L}^d_{i}(s)$, we use the binomial expansion for $(-1-e^{-4t})^{d-i}$ and $(-1-1)^{d-i}$ to rewrite $\mathcal{L}^d_{i}(s)$ as 
$$
\begin{aligned}
\mathcal{L}^d_{i}(s)&=-\int ^{\infty}_0((-1-e^{-4t})^{d-i} e^{-s^2t}e^{-2it}(I_0(2t))^i-(-2)^{d-i}e^{-t})\frac{\dt}{t}\\
&=-(-1)^{d-i}\sum_j\int ^{\infty}_0\dbinom{d-i}{j}(e^{-4t})^je^{-s^2t}e^{-2it}(I_0(2t)^i-e^{-t})\frac{\dt}{t}\\
&=-(-1)^{d-i}\sum_j\int ^{\infty}_0\dbinom{d-i}{j}\bigg\{e^{-s^2t}(e^{-4t})^je^{-2it}((I_0(2t)^i)-1)\\
&+\left(e^{-s^2t}(e^{-4t})^je^{-2it}-e^{-t}\right)\bigg\}\frac{\dt}{t}\\
&=-(-1)^{d-i}\sum_j\int ^{\infty}_0\dbinom{d-i}{j}e^{-s^2t}(e^{-4t})^je^{-2it}((I_0(2t)^i)-1)\frac{\dt}{t}\\
&+(-1)^{d-i}\sum_j\dbinom{d-i}{j}\log(s^2+2i+4j).
\end{aligned}
$$
From this calculation, we see that $$(-1)^{d-i}\sum_j\dbinom{d-i}{j}\log(s^2+2i+4j)\rightarrow 2^{d-i}\log(s^2)$$ and 
$$\int ^{\infty}_0e^{-s^2t}(e^{-4t})^je^{-2it}((I_0(2t)^i)-1)\frac{dt}{t}\rightarrow 0$$ as $s\rightarrow\infty$. This proves the asymptotic behavior of $\mathcal{L}^d_{i}(s)$. 
\end{proof}
\end{proposition}

\begin{proposition}\label{3.4}
Let
$$\mathcal{H}_N(s)=-\int^{\infty}_0\left\{e^{-s^2t}(\Theta_{L_{d}(N)}(t)-g(t))-(-1)^de^{-t}\right\}\frac{\dt}{t},$$
then $\mathcal{H}_N(s)$ satisfies the differential equation
$$\partial_s\mathcal{H}_N(s)=2s\int^{\infty}_0\left\{e^{-s^2t}(\Theta_{L_{d}(N)}(t)-g(t))\right\}\dt$$
and it is uniquely determined by the asymptotic behavior 
$$\mathcal{H}_N(s)=(-1)^d\log(s^2)+{o}(1)~~~as~~s\rightarrow\infty$$
\end{proposition}

\begin{proof}
We can show $\mathcal{H}_N(s)$ is differentiable following an argument similar to the proof of Theorem \ref{T3.1}. We can also easily verify that $\mathcal{H}_N(s)$ satisfies the associated first order ODE. Now, rewrite $\mathcal{H}_N(s)$ as follows:
\begin{align*}
 \mathcal{H}_N(s)=&-\int^{\infty}_0\left\{e^{-s^2t}(\Theta_{L_{d}(N)}(t)-g(t))-(-1)^de^{-t}\right\}\dfrac{\dt}{t}\\
=&-\int^{\infty}_0\left\{e^{-s^2t}\left(\Theta_{L_{d}(N)}(t)-g(t)-\bigg(\frac{-(1+e^{-4t})}{2}\bigg)^d\right)\right.\\
&\left.+\left(e^{-s^2t}\bigg(\frac{-(1+e^{-4t})}{2}\bigg)^d-(-1)^de^{-t}\right)\right\}\frac{\dt}{t}\\
=&-\int^{\infty}_0\left\{e^{-s^2t}\left(\Theta_{L_{d}(N)}(t)-g(t)-\bigg(\frac{-(1+e^{-4t})}{2}\bigg)^d\right)\right\}\frac{\dt}{t}+(-1)^d\log(s^2)
\end{align*}
In the last line, we used the fact that $\int_{0}^{\infty}(e^{-s^2t}-e^{-t})\frac{\dt}{t}=-\log(s^2)$ (see for example \cite{chinta2010zeta}). 
We observe using the proof of Lemma \ref{analysis of g},  
\begin{eqnarray*}
\int^{\infty}_0\left\{e^{-s^2t}\left(\Theta_{L_{d}(N)}(t)-g(t)-\bigg(\frac{-(1+e^{-4t})}{2}\bigg)^d\right)\right\}\frac{\dt}{t}= o(1)
\end{eqnarray*} as $s\to \infty$ and this completes the proof.
\end{proof}

From the definition of $V_{i}^{d,N}$, it follows that the following relation holds. 
\begin{lemma}\label{product to sum} The following identity holds:
$$\prod^d_{i=1}(n_i-1)=(-1)^d+\sum_{i=1}^{d}(-2)^{d-i}V_{i}^{d,N}.$$
\end{lemma} 

\begin{theorem}\label{T3.1}
For any $s\in\mathbb{C}$ with $Re(s^2)>0,$ we have 
$$\sum_{\lambda\in \Lambda_{\Ld}}\log(\lambda+s^2)=\sum^d_{i=1}V^{d, N}_i\mathcal{L}^d_{i}(s)+\mathcal{H}_N(s).$$
When $s\to 0$, we will have the identity
$$\sum_{\lambda\in \Lambda_{\Ld}}\log(\lambda)=\sum^d_{i=1}V^{d,N}_i\mathcal{L}^d_{i}(0)+\mathcal{H}_N(0)$$
where
$$
\mathcal{L}^d_{i}(s)=
-\int ^{\infty}_0((-1-e^{-4t})^{d-i} e^{-s^2t}e^{-2it}(I_0(2t))^i-(-2)^{d-i}e^{-t})\frac{\dt}{t}\quad 0<i\leq d,
$$
and
$$\mathcal{H}_N(s)=-\int^{\infty}_0\left\{e^{-s^2t}(\Theta_{L_{d}(N)}(t)-g(t))-(-1)^de^{-t}\right\}\frac{\dt}{t}$$ as in the Proposition \ref{3.3} and \ref{3.4} respectively. 
\end{theorem}

\begin{proof}
 Using differential equations in Lemma \ref{3.1}, Lemma \ref{3.2}, Proposition \ref{3.3}, and Proposition \ref{3.4}, we observe that
\begin{eqnarray}\label{factorization of discrete determinant}
\sum_{\lambda \in \Lambda_{\Ld}}\log(\lambda+s^2)=\sum^d_{i=1}V^{d,N}_i\mathcal{L}^d_{i}(s)+\mathcal{H}_N(s)+{C},
\end{eqnarray}
where ${C}$ is a constant. Using the asymptotic behavior as $s\to \infty$ from Lemma \ref{3.1} Lemma \ref{3.2}, Proposition \ref{3.3}, and Proposition \ref{3.4} and then using Lemma \ref{product to sum}, we conclude that $C=0$.

In order to complete the proof of the theorem, we need to study the relation (\ref{factorization of discrete determinant}) as $s\to 0$. In particular, we need to show the right hand side of (\ref{factorization of discrete determinant}) is finite when $s\to 0$. 

Using the facts (see for example \cite{chinta2010zeta}),
$$\begin{aligned}
e^{-2dt}(I_0(2t))^d-e^{-t}=\mathcal{O}(t)~~~&as~t\rightarrow0\\
e^{-2dt}(I_0(2t))^d=\mathcal{O}(t^{-d/2})~~~&as~t\rightarrow\infty
\end{aligned}$$
we see that the integrand in the definition of $\mathcal{L}^d_{d}(0)$ is in $L^1(0,\infty)$, w.r.t. $dt/t$. Hence, $\mathcal{L}^d_{d}(0)$ exists. When $i<d$,
$$\begin{aligned}
(1+e^{-4t})^{d-i} e^{-2it}(I_0(2t))^i)-2^{d-i}e^{-t}\approx2^{d-i}(e^{-2it}I_0(2t)^{i}-e^{-t})= \mathcal{O}(t)~~~&as~t\rightarrow0\\
(1+e^{-4t})^{d-i} e^{-2it}(I_0(2t))^i)-2^{d-i}e^{-t}\approx e^{-2it}(I_0(2t))^i)=\mathcal{O}(t^{-i/2})~~~&as~t\rightarrow\infty.
\end{aligned}$$

Using these relations, we find that the integrand in $\mathcal{L}^d_{i}(0)$ is $L^1(0,\infty)$, w.r.t. $\dt/t$.

Next, we analyze $\mathcal{H}_N(0)$. From the Lemma \ref{analysis of g}, we have 
$$\Theta_{L_{d}(N)}(t)-g(t)-(-1)^de^{-t}=o(t)~~~as~t\rightarrow0$$ 

Furthermore, as $t\to\infty$, we have
$$\Theta_{L_{d}(N)}(t)=\mathcal{O}(e^{-ct})~~~\text{for some}~c>0$$
and 
$$g(t)+(-1)^de^{-t}=\mathcal{O}(t^{-{1/2}}).$$
This means that $\mathcal{H}_N(0)$ is $L^1(0,\infty)$, w.r.t. $\dt/t$. Now, we can complete the proof by letting the limit $s\to 0$ in (\ref{factorization of discrete determinant}).
\end{proof}

\subsubsection{\textbf{Asymptotic behavior of determinants of discrete Laplacians}}
We recall that discrete hypercubes $L_{d}(N)$ and discrete tori $DT_d(N)$ also depend on another parameter $u$ which we may think as reciprocal of the lattice spacing. Hence, the various functions such as the discrete theta function and the determinant of discrete Laplacians are functions of $u.$ In this section, we will study the behavior
of these functions as $u\to\infty$. More precisely, we will prove one of the main results of this paper which is to obtain the asymptotic expansion of logarithm of the determinant of discrete Laplacian of hypercube as $u\to\infty$. 

\begin{proposition}\label{4.1}
 For each fixed $t>0$ we have the limit
 $$\Theta_{L_{d}(N)}(u^2t)\to\theta_{K_{d}(A)}(t)$$ as $u\to \infty.$
\end{proposition}

\begin{proof}
In the light of Proposition \ref{proposition2.1} and Proposition \ref{prop22.2} and the fact that 
$$(-1-e^{-4u^2t})\rightarrow-1~~~as~u\rightarrow\infty,$$ it is sufficient to show
$$\Theta_{DT_{p}(n_{i_1},\dots, n_{i_p})}(u^2t)\rightarrow\theta_{T_{p}(a_{i_1},...,a_{i_p})}(t)$$ as $u\to \infty.$
But this statement is already proved in \cite{chinta2010zeta}.
\end{proof}

\begin{proposition}\label{4.2}
As $u\to\infty$,
$$\begin{aligned}
\int^{\infty}_1&\left\{(\Theta_{L_{d}(N)}(u^2t)-g(u^2t))-(-1)^de^{-u^2t}\right\}\frac{\dt}{t}\\
&=\int^{\infty}_1(\theta_{K_{d}(A)}(t))\frac{\dt}{t}-\sum_{i\neq0}(-1)^{d-i}\frac{2}{i}V^d_i(4\pi)^{-i/2}+\mathcal{O}(1).
\end{aligned}$$
\end{proposition}

\begin{proof} Let us write
\begin{eqnarray}
\begin{array}{lll}
&\int^{\infty}_1\left\{(\Theta_{L_{d}(N)}(u^2t)-g(u^2t))-(-1)^de^{-u^2t}\right\}\dfrac{\dt}{t}\\
&=\int^{\infty}_1(\Theta_{L_{d}(N)}(u^2t))\dfrac{\dt}{t}-\int^{\infty}_1g(u^2t)\dfrac{\dt}{t}-1\int^{\infty}_1(-1)^de^{-u^2t}\dfrac{\dt}{t}.\\
\end{array}
\end{eqnarray}
Using Proposition \ref{4.1}, 
\begin{eqnarray}\int^{\infty}_1(\Theta_{L_{d}(N)}(u^2t))\frac{\dt}{t}=\int^{\infty}_1\theta_{K_{d}(A)}(t)\frac{\dt}{t}
\end{eqnarray} as $u\to\infty$. Note that
\begin{eqnarray*}
\hskip5mm V^{d,N}_d\int^{\infty}_1(e^{-u^2t}I_0({2u^2t})^d\frac{\dt}{t}\rightarrow V^d_d\int^{\infty}_1(4\pi t)^{-d/2}\frac{\dt}{t}=\frac{2}{d}V^d_d(4\pi)^{-d/2}
\end{eqnarray*}
when $u\rightarrow\infty$, which immediately implies the following: 
\begin{eqnarray}
\int^{\infty}_1g(u^2t)\frac{\dt}{t}=\sum_{i\neq0}(-1)^{d-i}\frac{2}{i}V^d_{i}(4\pi)^{-i/2}
\end{eqnarray} as $u\rightarrow\infty$.
Also, 
\begin{eqnarray}
\int^{\infty}_1e^{-u^2t}\frac{\dt}{t}=o(1)
\end{eqnarray} as $u\rightarrow\infty$. Using $(2.9), (2.10)$ and $(2.11)$ in $(2.8)$ proves the proposition.
\end{proof}

\begin{proposition}\label{4.3} For $u\to\infty$, the following holds :
$$\int^{1}_0\left\{\Theta_{L_{d}(N)}(u^2t)-g(u^2t)-\frac{1}{2^d}(-1-e^{-4u^2t})^d\right\}\frac{\dt}{t}\rightarrow\int^{1}_0\left\{\theta_{K_{d}(A)}(t)-f(t)\right\}\frac{\dt}{t}.$$
\begin{proof}
For fixed $t<1$, 
$$\Theta_{L_{d}(N)}(u^2t)-g(u^2t)\rightarrow\theta_{L_{d}(A)}(t)-f(t)+(-1)^d\frac{1}{2^d}~~~as~u\to\infty.$$
We also have
$$(-1-e^{-4u^2t})^d\rightarrow (-1)^d~~~as~u\rightarrow\infty.$$ These relations immediately imply the proposition.
\end{proof}
\end{proposition}

We recall the following well known fact from \cite{chinta2010zeta}:
\begin{equation}\label{Integralcomputation}
\int^1_0(e^{-u^2t}-1)\frac{\dt}{t}=\Gamma'(1)-\log(u^2)+o(1)~~~as~~~u\rightarrow\infty.
\end{equation}
We will use this fact in the next proposition to prove a modified version of this fact needed for our purpose.

\begin{proposition}\label{4.4} As $u\rightarrow\infty$,
$$\int^1_0\left((1+e^{-4u^2t})^d-2^de^{-t}\right)\frac{\dt}{t}=\log(u^2)-\sum^d_{i=1}\log(4i)\dbinom{d}{i}-\Gamma'(1)+o(1).$$

\begin{proof}
$$
\begin{aligned}
\int^1_0\left((1+e^{-4u^2t})^d-2^de^{-u^2t}\right)\frac{\dt}{t}&=\int^1_0\left((1+e^{-4u^2t})^d-2^d\right)\frac{\dt}{t}\\
&+\int^1_0\left(2^d-2^de^{-u^2t}\right)\frac{\dt}{t}
\end{aligned}
$$
Using (\ref{Integralcomputation}), 
$$\int^1_0\left(2^d-2^de^{-u^2t}\right)\frac{\dt}{t}=2^d(-\Gamma'(1)+\log(u^2))+o(1)$$
as $u\rightarrow\infty$.

Moreover, 
$$
\begin{aligned}
\int^1_0\left((1+e^{-4u^2t})^d-2^d\right)&=\sum_{i=0}^{d}\int^{1}_0\dbinom{d}{i}(e^{-4iu^2t}-1)\frac{\dt}{t}\\
&=\sum_{i=1}^{d}\dbinom{d}{i}(\Gamma'(1)-\log(4iu^2))+o(1)
\end{aligned}$$
as $u\rightarrow\infty$. Now, the proposition follows from combining these two observations.

\begin{comment}
Now,
\begin{eqnarray}
\begin{aligned}
&d\int^{4u^2}_{0}\log{4u^2}e^{-v}(1-e^{-v})^{d-1}dv\\
&=d\log(4u^2)\int^{4u^2}_{0}(1-e^{-v})^{d-1}d(1-e^{-v})\\
%&=\log(4u^2)(1-e^{-v})^d|^{4u^2}_0\\
&=\log(4u^2)(1-e^{-4u^2})^d \\
&=\log(4u^2) +o(1)
\end{aligned}
\end{eqnarray} as $u\to\infty.$

Also, using the binomial expansion
$$-d\int^{4u^2}_{0}\log(v)e^{-v}(1-e^{-v})^{d-1}dv=-d\sum_{i=0}^{d-1}\int^{4u^2}_0\log(v)e^{-(i+1)v}(-1)^i\dbinom{d-1}{i}dv$$
By the change of variable $w=(i+1)v$, we observe
$$
\begin{aligned}
\int^{4u^2}_0\log(v)e^{(i+1)v}dv&=\frac{1}{1+i}\int^{4u^2(1+i)}_0\log(w/i+1)e^{-w}dw\\
&=\frac{1}{1+i}\int^{4u^2(1+i)}_0\log(w)e^{-w}dw
-\frac{\log(i+1)}{1+i}\int^{4u^2(1+i)}_0e^{-w}dw.
\end{aligned}
$$
Note that
$$ \frac{\log(i+1)}{1+i}\int^{4u^2(1+i)}_0e^{-w}dw=\frac{\log(1+i)}{1+i}+o(1)~~~as~u\rightarrow\infty$$ and it is known \cite{chinta2010zeta} that
$$\frac{1}{1+i}\int^{4u^2(1+i)}_0\log(w)e^{-w}dw=\frac{\Gamma'(1)}{i+1}+o(1)$$ as $u\to \infty.$
Hence,
\begin{eqnarray}
\begin{array}{lll}
&-d\int^{4u^2}_{0}\log(v)e^{-v}(1-e^{-v})^{d-1}dv\\
&=-d\sum_{i=0}^{d-1}\int^{4u^2}_0\log(v)e^{-(i+1)v}(-1)^i\dbinom{d-1}{i}dv\\
&=-d\sum_{i=0}^{d-1}(-1)^{i}\dbinom{d-1}{i}\left(-\frac{\log(i+1)}{1+i}+\frac{\Gamma'(1)}{i+1}\right)+o(1)\\
&=-\sum_{i=1}^{d}(-1)^{i}\dbinom{d}{i}\log(i)-\Gamma'(1)+o(1)
\end{array}
\end{eqnarray} as $u\to\infty.$
\end{comment}
\end{proof}
\end{proposition}

We have the following corollary of this proposition.
\begin{corollary}\label{Cor2.3} As $u\rightarrow\infty$,
$$
\begin{aligned}
\int^{1}_0&\left\{\Theta_{L_{d}(N)}(u^2t)-g(u^2t)-(-1)^de^{-u^2t}\right\}\frac{\dt}{t}=\\ &\int^{1}_0\left\{\theta_{L(a_1,...,a_d)}(t)-f(t)\right\}\frac{\dt}{t}-\frac{(-1)^d}{2^d}(\Gamma'(1)-\log(u^2))\\
&-\frac{(-1)^d}{2^d}\sum^d_{i=1}\log(4i)\dbinom{d}{i}+o(1).
\end{aligned}
$$
\end{corollary}

\begin{proposition}
As $u\to \infty,$
$$
\mathcal{H}_N(u)(0)=-\frac{(-1)^d}{2^d}\log(u^2)-\zeta_{K_d(A)}'(0)+\frac{(-1)^d}{2^d}\sum^d_{i=1}\log(4i)\dbinom{d}{i}+o(1).$$
\begin{proof}
Combining the results of Proposition \ref{4.2} and Corollary \ref{Cor2.3}, we have
$$
\begin{aligned}
\mathcal{H}_N(u)(0)=&-\int^{1}_0\left\{\theta_{L(a_1,...,a_d)}(t)-f(t)\right\}\frac{\dt}{t}+\frac{(-1)^d}{2^d}(\Gamma'(1)-\log(u^2))\\
&-\int^{\infty}_1(\theta_{K_{d}(A)}(t))\frac{\dt}{t}+\sum_{i\neq0}(-1)^{d-i}\frac{2}{i}V^d_i(4\pi)^{-i/2}\\
&+\frac{(-1)^d}{2^d}\sum^d_{i=1}\log(4i)\dbinom{d}{i}+o(1).
\end{aligned}
$$ 
But notice that the expression in the first two lines is $\zeta'_{K_d(A)}(0)$ which completes the proof of this proposition.
\end{proof}
\end{proposition}

In summary, we have proved the following theorem which is the main result of this section:

\begin{theorem}\label{Th1.1}
Let $\Delta_{L_{d}(N)}$ be the discrete Laplacian on the $d$-dimensional discrete hypercube $L_{d}(N)$ with Dirichlet boundary condition.Let $\Delta_{K_{d}(A)}$ be the Laplacian the hypercube $K_{d}(A)$ with Dirichlet boundary condition. Then,
\begin{equation}
\begin{aligned}
 \log \det\Delta_{L_{d}(N)}&=\sum^d_{i=1}V^{d,N}_{i}\mathcal{L}^d_{i}(0)-\frac{(-1)^d}{2^d}\log(u^2)+\log {\det}_{\zeta}\Delta_{K_{d}(A)}\\
 &+\frac{(-1)^d}{2^d}\sum^d_{i=1}\log(4i)\dbinom{d}{i}+o(1)
 \end{aligned}
\end{equation}
as $u\rightarrow\infty$
where
$$
\mathcal{L}^d_{i}(s)=\left\{
\begin{aligned}
&-\int ^{\infty}_0(e^{-s^2t}e^{-2dt}(I_0(2t))^d-e^{-t})\frac{\dt}{t}~~~&i=d,\\
&-\int ^{\infty}_0((-1-e^{-4t})^{d-i} e^{-s^2t}e^{-2it}(I_0(2t))^i-(-2)^{d-i}e^{-t})\frac{\dt}{t}~~~&0<i<d,
\end{aligned}\right.
$$ \text{and} $\log {\det}_{\zeta}\Delta_{K_{d}(A)}$ is the zeta regularized determinant.
\end{theorem}

\begin{example} Consider the two dimensional rectangle with size $K_{2}(a_1,a_2)$ and $(n_1,n_2)=(na_1,na_2)$, then 
\begin{eqnarray*}
&\log \det\Delta_{L_{2}((n_1,n_2)}=\frac{4G}{\pi}a_1a_2n^2-\log(1+\sqrt{2})(a_1n+a_2n)-1/2\log n\\
&+\log {\det}_{\zeta}\Delta_{K_{2}((a_1,a_2)}-\frac{1}{4}\log(2)
\end{eqnarray*}
where $G$ is the Euler's constant. This result agrees with the known result in \cite{duplanticr1988exact}.
\end{example}

\begin{remark}
In \cite{louis2017asymptotics}, similar problem for the free boundary condition is studied following a different approach than in this paper. Our consideration for the Dirichlet boundary condition is inspired by the desire to understand the partition function on a manifold with boundary. 
\end{remark}
 
 %%%%%%%%%%%%%%%%%%%%%%%%%%%%%%%%%%%%%%%%%%%%%%%%%%%%%%%%%%%%%%%%%%%%%%%%%%%%%%%%%%

\section{Hadamard partie finie regularization on hypercube}
In this section, we modify the techniques developed in \cite{vertman2015regularized} to obtain a relation between the regularized limit of the log-determinant of discrete Laplacian on $L_{d}(N)$ and the log of the zeta regularized determinant of the Laplacian on $K_{d}(A)$ for the Dirichlet boundary condition. Roughly speaking, it goes as follows. We first show that the resolvent trace of the discrete Laplacian with free boundary condition admits a polyhomogeneous expansion, and then calculate the regularized limit. Next, we use the relation between of the spectrum of discrete Laplacian on hypercube with Dirichlet and free boundary condition to relate the corresponding resolvent trace and calculate the regularized limit of the former. This will allow us to establish the relation between of the zeta regularized determinant and regularized limit of the log-determinant of discrete Laplacian for the Dirichlet boundary condition.

\subsection{Regularized limit and regularized integral}
We first recall few definitions which are taken from \cite{vertman2015regularized}. 

Let $E$ be an index set which is a subset of $\mathbb{C}\times \mathbb{N}$ with the property that
\begin{equation}
\{(\alpha,k)\in E| Re(\alpha)\geq s\} \text{ is a finite set for any $s\in\mathbb{R}$.}
\end{equation}

\begin{definition}
A function $f\in C^\infty(\mathbb{R}_+,\mathbb{C})$ is a polyhomogeneous function with respect to $x\rightarrow \infty$, if there exists an index set $E$, such that $f$ has an asymptotic expansion of the form,
\begin{equation}\label{expansion}
f\sim \sum_{(\alpha,k)\in E}a_{\alpha,k} x^{\alpha}\log^k(x)\qquad as~x\rightarrow \infty
\end{equation}
where $a_{\alpha,k}\in\mathbb{C}$, and $\mathbb{R}_+=(0,\infty)$.
\end{definition}
We can also define the polyhomogeneous function with respect to $x\rightarrow 0$. In this case, we choose an index set $\tilde{E}\subset\mathbb{C}\times\mathbb{N}$ with the property,
\begin{equation}
\{(\alpha,k)\in \tilde{E}| Re(\alpha)\leq s\} \text{ is a finite set for any $s\in\mathbb{R}$}
\end{equation} and demand the asymptotic expansion of the form (\ref{expansion}) for $x\to 0$.

\begin{definition}
Let $f$ be a function which is a polyhomogeneous with respect to $x\to\infty$ the \emph{regularized limit} of $f(x)$ is defined by,
\begin{align*}
\LIM_{x\to \infty} f(x) := a_{00}.
\end{align*}
Similarly, if $f$ is a polyhomogeneuos with respect to $x\to0$, then the \emph{regularized limit} of $f(x)$ is defined by
\begin{align*}
\LIM_{x\to 0} f(x) := a_{00}.
\end{align*}
\end{definition}

Using this notion of regularized limit we can define the regularized integral $\avint$ of a function $f$ as follows:
\begin{eqnarray}
\avint^{\infty}_1f(x)\, \d x:=\LIM_{R\rightarrow\infty}\int^R_1f(x)\,\d x \hskip2mm\text{and}\hskip2mm \avint^{1}_{0}f(x)\, \d x:=\LIM_{\varepsilon\rightarrow 0}\int_{\varepsilon}^1f(x)\, \d x.
\end{eqnarray}
If both regularized integrals of $f$ exists, then we define 
$$\avint_{0}^{\infty}f(x)\, \d x= \avint^{\infty}_1f(x)\, \d x+\avint^{1}_{0}f(x)\, \d x. $$

We will need the following definition as well.
\begin{definition}
A function $f\in C^\infty(\mathbb{R}_+^2,\mathbb{C})$ is said to admit a partial polyhomogeneous expansion, if there exists index sets $E,E'$ and a constant $N\in\mathbb{Z}$, such that,
\begin{equation}
f(z,n)\sim \sum_{l\geq N}f_{(l)}(z,n),
\end{equation}
where each is homogeneous of degree $l$ jointly in $(z,n)$. With the asymptotic expansion of $f(z,1)$ and $f(1,n)$ are both polyhomogeneous functions as $z\rightarrow \infty$ or $n\rightarrow \infty$, with the index sets $E,E'$. 
\end{definition}

\subsection{Resolvent trace and polyhomogeneous expansion}

Let $\Delta_{L_d(N)}^F$ denote the discrete Laplacian on the discrete hypercube $L_{d}(n,\dots,n)$ with the free boundary condition. Later we will consider Dirichlet boundary condition. Let $\tilde{\Delta}_{L_d(N)}^F$ denote the operator $n^2\Delta_{L_d(N)}^F$. We will call $\Tr(\tilde{\Delta}_{L_d(N)}^F+z^2)^{-d}$ the \textit{resolvent trace} of $\tilde{\Delta}_{L_d(N)}^F$. Let us note that the eigenvalues of $\tilde{\Delta}_{L_d(N)}^F$ are given by 
$$
4n^2\sum_{i=1}^d\sin^2\left(\frac{\pi k_i}{2n}\right)
$$
 where $k_i\in\{1,2,\dots,n-1\}.$ For $x=(x_1,...,x_m)$, 
$$
\omega(n,x):=4n^2\sum_{i=1}^d\sin^2\left(\frac{\pi x_i}{2n}\right).
$$
Given $k$-tuple $J=(j_1,...,j_k)$ of distinct integers with $k\leq d$, we define 
$$
\{x_J=n\}:=\{x\in\mathbb{N}^d_0\cap[0,n]^d|x_{j_i}=...=x_{j_k}=n\}
$$
With this notation at hand, the resolvent trace can be written as follows:
$$
\Tr(\tilde{\Delta}_{L_d(N)}^F+z^2)^{-d}=\sum^d_{k=0}(-1)^k\sum_{|J|=k}\sum_{\{x_J=n\}}(\omega(n,x)+z^2)^{-d}.
$$
Define $S(z,n)$ by 
$$
S(z,n) = \sum^n_{x_1=0}...\sum^n_{x_d=0} \left(\omega(n,x_1,...,x_d) + z^2\right)^{-d}.
$$
Using the Euler Maclaurin summation formula iteratively, we can rewrite $S(z,n)$ as follows.
\begin{lemma} For any $M\in \mathbb{N}$, we have the following identity.
$$\label{S}
S(z,n) = \sum_{\beta\in \{1,2,3,4\}^d} P_{\beta_1,1}
\circ \cdots \circ P_{\beta_d,d} \left(\omega(n,x_1,...,x_d) + z^2\right)^{-d},
$$
where each $P_{\beta_j,j}$ acts in the $x_j$-variable on
$u\in C^\infty[0,\infty)$ by
$$
P_{\beta_j,j} u := \left\{
\begin{aligned}
&\int_0^n u(x_j) \d x_j, &\ \textup{if} \ \beta_j=1, \\
&\sum_{k=1}^M \frac{B_{2k}}{(2k)!} \left(\partial^{(2k-1)}_{x_j}|_{x_j=n} - \partial^{(2k-1)}_{x_j}|_{x_j=0}\right)u, &\ \textup{if} \ \beta_j=2, \\
&\frac{1}{(2M+1)!} \int_0^n B_{2M+1} (x_j - [x_j]) \partial^{(2M+1)}_{x_j} u(x_j) \d x_j, &\ \textup{if} \ \beta_j=3, \\
&\frac{1}{2}\left(u(x_j=n) + u(x_j=0)\right), &\ \textup{if} \ \beta_j=4.
\end{aligned} \right.
$$
Here, $B_i(x)$ is the $i$-th Bernoulli polynomial and $B_i$ is the $i$-th Bernoulli number.
\end{lemma}

We note that the identity in the lemma depends on the choice of $M\in\mathbb{N}$ which we will always choose to be sufficiently large. This lemma plays an important role in the analysis of $S(z,n)$. With the help of this lemma and following the ideas in \cite{vertman2015regularized}, we will show that $S(z,n)$ admits a polyhomogeneous expansion. More precisely, we have the following proposition.

\begin{proposition}\label{1.2.2.15}
The function $S(n,z)$ admits a partial polyhomogeneous expansion
$$S(n,z) = \sum_{j=0}^{d} h'_{-d-j}(z,n) + H'(z,n),$$
where each $h'_{-d-j}\in C^\infty(\mathbb{R}_+^2)$ is a partial polyhomogeneous function, and it is homogeneous of order $(-d-j)$ jointly in $(z,n)$.The remainder term, which depends on choice of $M$ in the Lemma \ref{S}, satisfies $H'_N(z,n) = O(z^{-2d-2})$, as $z\to \infty$,
uniformly in $n>0$. Furthermore,
$$
h^{'}_{-2d}(z,n) = \frac{1}{2^d}\sum^d_{k=0} \dbinom{d}{k}\frac{1}{(z^2+4kn^2)^d}.
$$
\begin{proof}
The proof of the first part of the proposition is essentially the same as of Proposition 3.2 in \cite{vertman2015regularized}. The only difference is the computation of the term $h^{'}_{-2d}(z,n)$ which is homogeneous of degree -2d. Note that 
$$P_{4,j}(4n^2\sin^2{\frac{\pi x_j}{n}})=\frac{1}{2}(4n^2+0)).$$ 
Using this, an inductive calculation will give 
$$\begin{aligned}
h^{'}_{-2d}(z,n)
=&P_{4,1}\circ...\circ P_{4,d}(\omega(n,x)+z^2)^{-d})\\
=& \frac{1}{2^d}\sum^d_{k=0} \dbinom{d}{k}\frac{1}{(z^2+4kn^2)^d}.
\end{aligned}
$$
\end{proof}
\end{proposition}

Next, we analyze the behavior of the resolvent trace which will be the key to prove the main result of this section.

\begin{proposition}
The resolvent trace admits a partial polyhomogeneous expansion
$$\Tr(\tilde{\Delta}_{L_d(N)}^F+z^2)^{-d} = \sum_{j=0}^{d} h_{-d-j}(z,n) + H(z,n),$$
where each $h_{-d-j}\in C^\infty(\mathbb{R}_+^2)$ is homogeneous of order $(-d-j)$ jointly in $(z,n)$,
$h_{-d-j}(z, 1)$ and $h_{-d-j}(1, n)$ admit an asymptotic expansion of the form
\eqref{expansion} as $z, n\to \infty$, respectively.
The remainder term satisfies $H_N(z,n) = O(z^{-2d-2})$, as $z\to \infty$,
uniformly in $n>0$. Moreover, 
$$
h_{-2d}(z,n) = \sum_j(-1)^j\dbinom{d}{j}\frac{1}{2^d}\frac{1}{(z^2+4n^2j)^d}.\label{eq.2.15}
$$
\begin{proof}
The proof of this proposition follows the arguments of Theorem 3.3 of \cite{vertman2015regularized} except for the expression of 
$h_{-2d}(z,n)$. Using the argument in Proposition \ref{1.2.2.15}, we have 
$$
\begin{aligned}
h_{-2d}(z,n)=&\sum^d_{k=0}(-1)^k\sum_{|J|=k}\frac{1}{2^{d-k}}\sum^{d-k}_{j=0} \dbinom{d-k}{j}\frac{1}{(z^2+4(j+k)n^2)^d}\\
=&\sum^d_{k=0}(-1)^k\dbinom{d}{k}\frac{1}{2^{d-k}}\sum^{d-k}_{j=0} \dbinom{d-k}{j}\frac{1}{(z^2+4(j+k)n^2)^d}\\
=&\sum^d_{k=0}(-1)^k\dbinom{d}{k}\frac{1}{2^{d-k}}\sum^{d}_{j=k} \dbinom{d-k}{j-k}\frac{1}{(z^2+4n^2j)^d}\\
\end{aligned}
$$
 rearranging the order of summation
$$
\begin{aligned}
=&\sum^{d}_{j=0}\sum^j_{k=0}(-1)^k\dbinom{d}{k}\frac{1}{2^{d-k}} \dbinom{d-k}{j-k}\frac{1}{(z^2+4n^2j)^d}\\
=&\sum^{d}_{j=0}\frac{1}{2^d}\dbinom{d}{j}\sum^j_{k=0}(-2)^k\dbinom{j}{k}\frac{1}{(z^2+4n^2j)^d}\\
=&\sum^d_{j=0}(-1)^j\dbinom{d}{j}\frac{1}{2^d}\frac{1}{(z^2+4n^2j)^d}
\end{aligned}
$$
\end{proof}

\end{proposition}
The next lemma will be useful for computation the regularized integral of $h_{-2d}(z,1)$. 

\begin{lemma}\label{regularized integral of h} For any positive integer $d$ and positive $\lambda$ the following holds:
$$
-2\avint^{\infty}_{0}\frac{z^{2d-1}}{(z^2+\lambda)^d}\d z=\log\lambda.
$$
\begin{proof} A simple computation shows 
$$-2\avint^{\infty}_{0}\frac{z}{(z^2+\lambda)}\d z=\log\lambda$$ and now the lemma follows from the following identity from section 1.3 of \cite{vertman2015regularized}:
$$
-2\avint^{\infty}_{0}z^{2d-1}\frac{1}{(z^2+\lambda)^d}\d z=-2\avint^{\infty}_{0}\frac{z}{(z^2+\lambda)}\d z.
$$

\end{proof}
\end{lemma}

We recall from \cite{louis2017asymptotics} that $\log\det\tilde{\Delta}_{L_d(N)}^F$ admits an asymptotic expansion as $n\to\infty$. Let us denote the constant term in this asymptotic expansion by $\log \det\bar{\Delta}_{L_d(1,\dots,1)}^F$. Now, we are ready to state the main result of this section.

\begin{theorem}\label{th3.6}
The regularized limit of $\log\det\tilde{\Delta}_{L_d(N)}^F$ as $n\rightarrow\infty$ exists. Moreover, we have
$$
\log \det\bar{\Delta}_{L_d(1,\dots,1)}^F=\LIM_{n\rightarrow\infty}\log \det\tilde{\Delta}_{L_d(N)}^F-\frac{1}{2^d}\sum^d_{j=1}\log(4j)(-1)^j\dbinom{d}{j}.
$$
\end{theorem}

\begin{proof}
The proof is essentially same as proof Theorem 3.3 in \cite{vertman2015regularized}. More precisely, we have
$$
\log \det\bar{\Delta}_{L_d(1,\dots,1)}^F=\LIM_{n\rightarrow\infty}\log \det\tilde{\Delta}_{L_d(N)}^F+2\avint^{\infty}_{0}z^{2d-1}h_{-2d}(z,1)\d z.
$$
Using the previous lemma, we compute
$$
-2\avint^{\infty}_{0}z^{2d-1}h_{-2d}(z,1)\d z=\frac{1}{2^d}\sum^d_{j=1}\log(4j)(-1)^j\dbinom{d}{j}
$$ which completes the proof.
\end{proof}

Now, we consider the Dirichlet boundary condition. The relationship between the eigenvalues for Dirichlet boundary condition the eigenvalues of the free boundary condition discussed in section 2 allows us to express the regularized limit of log-determinant of Laplacian on discrete hypercubes with Dirichlet boundary condition in terms of free boundary condition. More precisely, we have the following. 
\begin{proposition}\label{prop3.2}
Let $\tilde{\Delta}^D_{L_d(N)}$, and $\tilde{\Delta}^F_{L_d(N)}$ denote the discrete Laplacian on a discrete hypercube with Dirichlet boundary condition and free boundary condition, we have 
$$
\LIM_{n\rightarrow \infty} \log \det\tilde{\Delta}^D_{L_d(N)} = \sum^{d}_{i=1}(-1)^{d-i}\dbinom{d}{i}\LIM_{n\rightarrow \infty}\log \det \tilde{\Delta}^F_{L_d(N)}
$$
\end{proposition}
\begin{proof}
As $\LIM$ is linear, all we need to prove is the following:
$$
\log \det \Delta^D_{L_d(N)} = \sum^{d}_{i=1}(-1)^{d-i}\dbinom{d}{i}\log \det \Delta^F_{L_i(N)}
$$
which follows from the relationship between the eigenvalues. 
\end{proof}

On the continuum side, we have the following analogous relation:
\begin{proposition}\label{prop3.3}
Let $\Delta^D_{L_d(A)}$ denote the Laplacian on hypercube with Dirichlet boundary condition and the $\tilde{\Delta}^F_{L_d(A)}$ as defined in Theorem \ref{th3.6}, we have,
$$
\log {\det}_{\zeta} \Delta^D_{L_d(A)} = \sum^{d}_{i=1}(-1)^{d-i}\dbinom{d}{i}\log {\det}_{\zeta}\bar{\Delta}^F_{L_i(A)}.
$$
\end{proposition}
\begin{proof}
The same argument as in Proposition \ref{prop3.2}.
\end{proof}
Using Proposition \ref{prop3.2} and Proposition \ref{prop3.3} we have the following result. 
\begin{theorem} The following holds:
$$
\log {\det}_{\zeta} \Delta^D_{L_d(A)} =\LIM_{n\rightarrow \infty} \log \det \tilde{\Delta}^D_{L_d(N)}-(-1)^d\frac{1}{2^d}\sum^d_{i=1}\log(4i)\dbinom{d}{i}.
$$
\end{theorem}
\begin{proof}
We only need to show the second term in the formula above; Using Proposition \ref{prop3.2}, we have
$$
\begin{aligned}
\sum^{k}_{d=1}(-1)^{k-d}\dbinom{k}{d}&\frac{1}{2^d}\sum^d_{i=1}\log(4i)(-1)^i\dbinom{d}{i}\\
&=(-1)^{k}\sum^{k}_{d=1}\dbinom{k}{d}\frac{1}{2^d}\sum^d_{i=1}(-1)^{i-d}\log(4i)\dbinom{d}{d-i}
\end{aligned}
$$
changing the index i and d, we have,
$$
\begin{aligned}
(-1)^{k}\sum^{k}_{d=1}\dbinom{k}{d}&\frac{1}{2^d}\sum^d_{i=1}(-1)^{i-d}\log(4i)\dbinom{d}{d-i}\\
&=(-1)^{k}\sum^{k}_{i=1}\log(4i)\sum^k_{d=i}\dbinom{k}{d}\frac{1}{2^d}(-1)^{i-d}\dbinom{d}{d-i}\\
&=(-1)^{k}\sum^{k}_{i=1}\dbinom{k}{i}\log(4i)\sum^k_{d=i}\frac{1}{2^d}(-1)^{d-i}\dbinom{k-i}{d-i}\\
&=(-1)^{k}\sum^{k}_{i=1}\dbinom{k}{i}\log(4i)\sum^{k-i}_{d=0}\frac{1}{2^{d+i}}(-1)^{d}\dbinom{k-i}{d}\\
&=(-1)^{k}\sum^{k}_{i=1}\dbinom{k}{i}\log(4i)\frac{1}{2^i}\left(1-\frac{1}{2}\right)^{k-i}\\
&=(-1)^{k}\sum^{k}_{i=1}\dbinom{k}{i}\log(4i)\frac{1}{2^k}
\end{aligned}
$$
which finishes the proof.
\end{proof}
%%%%%%%%%%%%%%%%%%%%%%%%%%%%%%%%%%%%%%%%%%%%%%%%%%%%%%%%%%%%%%%%%%

\section{Massive Laplacian}
In this section we consider massive Laplacians and tori $T_{d}(A)$ and $DT_d(N)$ discrete tori. In the continuum case, the massive Laplacian $\Delta_{T_{d}(A)}+m^2$ where $m$ is a positive number, will be considered as an operator on the space of square integrable function on the torus with the space of smooth functions as the domain of the operator. 

In the discrete case, we need to be careful about the mass term if we want the limit as mesh approaches to zero to converge to the continuum massive Laplacian.
Let $\tilde{m}(u)$ be a positive function of $u$ such that
$$\lim_{u\to\infty}u\tilde{m}(u)=m.$$
The discrete massive Laplacian, as a linear operator on the finite dimensional vector space namely the space of functions on the discrete torus $DT_{d}(N)$, is defined as follows: $$\Delta^{\tilde{m}}_{DT_{d}(N)}f(x_j):=-\sum_j1/2(f(x_j+1/u)+f(x_j-1/u)-2f(x_j))+\tilde{m}(u)^2f(x_j).$$

\subsection{Zeta regularized determinant of massive Laplacian}
Let us briefly discuss the zeta regularized determinant of massive Laplacian on the torus $T_d(A)$. %For the notational convenience, in this section, we will simply write $\Delta$ for $\Delta_{T_{d}(A)}.$
\iffalse
The heat kernel $K(t,y)$ of the massive Laplacian is the solution of 
$$m^2K(t,y)+\Delta K(t,y)+\frac{\partial}{\partial t}K(t,y)=0.$$
\fi

Let us denote the spectrum of the massive Laplacian on the torus $T_{d}(A)$ by $\Lambda_{T_d(A)}^m.$ Since, the spectrum of the Laplacian on torus is known, we have 
$$\Lambda_{T_d(A)}^m=\{{m^2+(2\pi q_1)^2}/{a_1^2}+\cdots+{(2\pi q_d)^2}/{a_d^2}|(q_1,...,q_d)\in\mathbb{Z}^d\}.$$
Hence, the associated theta function is 
$$\theta_{T_d(A,m)}(t)=\sum_{\lambda\in\Lambda_{T_d(A)}^m}e^{-\lambda t}.$$
The asymptotic behavior of the theta function is as follows:
$$
\theta_{T_d(A,m)}(t)=\left\{\begin{aligned}
&V^d_d(4\pi t)^{-d/2}e^{-m^2t}+o(e^{-c/t})~~~&as~t\to0\\
&o(e^{-ct})~~~&as~t\to\infty
\end{aligned}
\right.
$$
here $c\in \mathbb{R}$, $c>0$. %, and \textcolor{red}{$V(A)$????} is the volume of the torus.

The zeta function associated to the massive Laplacian is defined as the inverse of the Mellin transformation of the theta function:
$$\zeta(s)= \frac{1}{\Gamma(s)}\int^{\infty}_{0}\theta_{T_{d}(A)}(m,t)t^s\frac{\dt}{t} $$
here the integral is well defined for $Re(s)>n/2$. Moreover, 
$$
\begin{aligned}
 \zeta(s)= &\frac{1}{\Gamma(s)}\int^{1}_{0}\big(\theta_{T_d(A,m)}(t)-V^d_de^{-m^2t}(4\pi t)^{-d/2}\big)t^s\frac{\dt}{t}\\
  &+\frac{1}{\Gamma(s)}\int^{1}_{0}V^d_de^{-m^2t}(4\pi t)^{-d/2}t^s\frac{\dt}{t}\\
  &+\frac{1}{\Gamma(s)}\int^{\infty}_{1}{\theta}_{T_d(A,m)}(t)t^s\frac{\dt}{t}.
\end{aligned}
$$
Hence, 
\begin{equation}\label{eq111}
\begin{aligned}
&-\log\det(\Delta_{T_d(A)}+m^2)=\zeta'(0)\\
&= \int^{1}_{0}({\theta}_{T_d(A,m)}(t)-V^d_de^{-m^2t}(4\pi t)^{-d/2})\frac{\dt}{t} +\int^{\infty}_{1}{\theta}_{T_d(A,m)}(t)\frac{\dt}{t}\\
   &+\dfrac{d}{ds}\bigg|_{s=0}\left(\frac{1}{\Gamma(s)}\int^1_0V^d_de^{-m^2t}(4\pi t)^{-d/2}t^{s-1}\dt\right).
\end{aligned}
\end{equation}

We will address the third term of (\ref{eq111}) in Proposition \ref{04.4}.

\subsection{Theta function and log determinant of massive discrete Laplacian}
The theta function of the massive discrete Laplacian is defined as
$$\Theta_{DT_d(N)}^{\tilde{m}}(t)=\sum_{\lambda\in\Lambda^{\tilde{m}}_{DT_d(N)}}e^{-\lambda t}$$
where $\Lambda^{\tilde{m}}_{DT_d(N)}$ is the spectrum of the discrete massive Laplacian and it is given by $\Lambda^{\tilde{m}}_{DT_d(N)}=\{2d+\tilde{m}^2-2\cos(2\pi q_1/n_1)...-2\cos(2\pi q_d/n_d)|0\leq q_i< n_i\}.$

Next we state two lemmas, whose proofs are essentially the same as proofs of Lemma \ref{3.1} and Lemma \ref{3.2} respectively. These lemmas will be useful in the analysis of logarithm of determinant of massive discrete Laplacian.
\begin{lemma}\label{2.1}
For all $s\in\mathbb{C}$ with $Re(s^2)>0$, we have
$$
\begin{aligned}
\sum_{\lambda\in\Lambda^{\tilde{m}}_{DT_d(N)}}\frac{2s}{s^2+\lambda}&=2sV^{d,N}_d\int^{\infty}_0e^{-s^2t}e^{-\tilde{m}^2t}e^{-2dt}I_0(2t)^d \dt\\
&+2s\int^{\infty}_0e^{-s^2t}[\Theta_{DT_d(N)}^{\tilde{m}}(t)-V^{d,N}_de^{-\tilde{m}^2t}e^{-2dt}I_0(2t)^d]\dt.
\end{aligned}$$
\end{lemma}

\begin{lemma}\label{2.2}
Let $$f_{\tilde{m}}(s)= \sum_{\lambda\in \Lambda^{\tilde{m}}_{DT_d(N)}}\log(\lambda+s^2)$$
then $f_{\tilde{m}}(s)$ is uniquely determined by the differential equation
$$\partial_sf_{\tilde{m}}(s)=\sum_{\lambda\in \Lambda^{\tilde{m}}_{DT_d(N)}}\frac{2s}{s^2+\lambda}$$
and the asymptotic behavior 
$$f_{\tilde{m}}(s)=V^{d,N}_d\log(s^2)+o(1)$$ as $s\to\infty$.
\end{lemma}

\begin{proposition}\label{p2.1}
Let 
$$\mathcal{L}_{\tilde{m}}(s)=-\int^{\infty}_0(e^{-s^2t}e^{-\tilde{m}^2t}e^{-2dt}I_0(2t)^d-e^{-t})\frac{\dt}{t}$$
then $\mathcal{L}_{\tilde{m}}$ is uniquely determined by the differential equation
$$\partial_s\mathcal{L}_{\tilde{m}}(s)=2s\int^{\infty}_{0} e^{-s^2t}e^{-\tilde{m}^2t}e^{-2dt}(I_0(2t))^d\,\dt$$
and the asymptotic behavior 
$$\mathcal{L}_{\tilde{m}}(s)=\log(s^2)+o(s)~~~as~s\to\infty.$$
\end{proposition}

\begin{proof} Notice that $\mathcal{L}_{\tilde{m}}(s)$ is bounded by $\mathcal{L}_{\tilde{m}=0}(s)$, which is the case considered in \cite{chinta2010zeta}. This means we can change the order of derivative and the integral here.

It is easy to check $\mathcal{L}_{\tilde{m}}$ satisfies the differential equation. To verify the asymptotic behavior, we first rewrite $\mathcal{L}_{\tilde{m}}$ as
$$
\begin{aligned}
\mathcal{L}_{\tilde{m}}(s)&=-\int ^{\infty}_0 e^{-s^2t}e^{-\tilde{m}^2t}e^{-2dt}((I_0(2t))^d-1)\frac{\dt}{t}+\log(s^2+\tilde{m}^2+2d).
\end{aligned}
$$
Now, the asymptotic behavior follows the observation that the first term approaches zero as $s\to\infty$ and the second term behaves as $\log(s^2)$ when $s\to\infty$. 
\end{proof}

We will also need the following proposition which is similar in spirit to Proposition 2.4.

\begin{proposition}\label{p2.2}
Let 
$$\mathcal{H}_N(s)=-\int^{\infty}_0e^{-s^2t}\big(\Theta_{DT_d(N)}^{\tilde{m}}(t)-V^{d,N}_de^{-\tilde{m}^2t}e^{-2dt}I_0(2t)^d\big)\frac{\dt}{t}$$
then $\mathcal{H}_N(s)$ is uniquely determined by the differential equation
$$\partial_s\mathcal{H}_N(s)=2s\int^{\infty}_0\left\{e^{-s^2t}\big(\Theta_{DT_d(N)}^{\tilde{m}}(t)-V^{d,N}_de^{-\tilde{m}^2t}e^{-2dt}I_0(2t)^d\big)\right\}\,\dt$$
and the asymptotic behavior 
$$\mathcal{H}_N(s)=o(1)~~~as~~s\rightarrow\infty.$$
\end{proposition}
Notice here that when $\tilde{m}=0$ we arrive at a case considered in \cite{chinta2010zeta} and as discussed before this allows us to change the order of derivative and integral here. 

The next theorem allows us to express the log determinant of the massive discrete Laplacian in terms of $\mathcal{L}_{\tilde{m}}(0)$ and $\mathcal{H}_N(0)$.

\begin{theorem}\label{T2.2}
For any $s\in\mathbb{C}$ with $Re(s^2)>0$ we have a relation
$$\sum_{\lambda\in\Lambda^{\tilde{m}}_{DT_d(N)}}\log(\lambda+s^2)=
V^{d,N}_d\mathcal{L}_{\tilde{m}}(s)+\mathcal{H}_N(s).$$
Moreover, letting $s\rightarrow0$ we have 
$$\sum_{\lambda\in\Lambda^{\tilde{m}}_{DT_d(N)}}\log\lambda=
V^{d,N}_d\mathcal{L}_{\tilde{m}}(0)+\mathcal{H}_N(0)$$
where
$$\mathcal{L}_{\tilde{m}}(0)=-\int^{\infty}_0(e^{-\tilde{m}^2t}e^{-2dt}I_0(2t)^d-e^{-t})\frac{\dt}{t}$$
and
$$\mathcal{H}_N(0)=-\int^{\infty}_0(\Theta_{DT_d(N)}^{\tilde{m}}(t)-V^{d,N}_de^{-\tilde{m}^2t}e^{-2dt}I_0(2t)^d)\frac{\dt}{t}.$$

\begin{proof}
Same as the proof of Theorem \ref{T3.1}. 
\end{proof}

\end{theorem}

\subsection{Asymptotic behavior of log determinant}
In this section, we analyze asymptotic behavior of log determinant of the discrete massive Laplacian as the parameter $u$ approaches infinity. 

The following proposition deals with the convergence of theta function of the discrete massive Laplacian.
\begin{proposition}\label{04.1}
 For each fixed $t>0$ we have the limit
 $$\lim_{u\to\infty}\Theta_{DT_d(N)}^{\tilde{m}}(u^2t)= \theta_{T_d(A,m)}(t).$$
\end{proposition}

\begin{proof}
The proposition follows from following the well known fact
$$\Theta_{DT_d(n_1,...,n_d)}(u^2t)\rightarrow\theta_{T_d(a_1,...,a_d)}(t)$$
and 
${u}^2\tilde{m}({u})^2\rightarrow m^2$ as $u\rightarrow\infty$.
\end{proof}

\begin{proposition}\label{04.2} We have the following limit:
\begin{equation}
\begin{aligned}\label{eq04.4}
\lim_{u\to\infty}\int^{\infty}_1&(\Theta_{DT_d(N)}^{\tilde{m}}(u^2t)-V^{d,N}_de^{-u^2\tilde{m}^2t}e^{-2du^2t}I_0(2u^2t)^d)\frac{\dt}{t}\\
&=\int^{\infty}_1\theta_{T_d(A,m)}(t)\frac{dt}{t}-\int^{\infty}_1V^d_d(4\pi t)^{-d/2}e^{-m^2t}\frac{\dt}{t}.
\end{aligned}
\end{equation}
\end{proposition}

Now, we can analyze the third term in (\ref{eq111}). We will consider this together with the second term of (\ref{eq04.4}) in the following proposition.

\begin{proposition}\label{04.4} We have the following:
$$\begin{aligned}
\frac{\d}{\d s}\bigg|_{s=0}
 &\left(\frac{1}{\Gamma(s)}\int^{\infty}_0V^d_d(4\pi t)^{-d/2}e^{-m^2t}t^s\frac{\dt}{t}\right)\\
&=\left\{
\begin{aligned}
&\frac{V^d_d}{(4\pi/m^2)^{d/2}}(m^2)^{s-1}\Gamma(-d/2)~~~~~\text{for}~~~~~d/2\not\in\mathbb{Z}\\
&(-1)^{d/2}
\frac{2/d+2/(d-2)\cdots+1-\log(m^2)}{(d/2)!}\frac{V^d_d}{(4\pi)^{d/2}}m^d~~\text{for}~d/2\in\mathbb{Z}.
\end{aligned}
\right.
\end{aligned}
$$
\begin{proof}
Notice that we can rewrite the second part of Proposition \ref{04.2}
$$
\int^{\infty}_1V^d_d(4\pi t)^{-d/2}e^{-m^2t}\frac{\dt}{t}
=\frac{\d}{\d s}\bigg|_{s=0}\left(
\frac{1}{\Gamma(s)}\int^{\infty}_1V^d_d(4\pi t)^{-d/2}e^{-m^2t}t^s\frac{\dt}{t}
\right).
$$
Adding this with the third term of (\ref{eq111}), we have the following expression
\begin{equation}\label{eq10}
\frac{\d}{\d s}\bigg|_{s=0}\left(
\frac{1}{\Gamma(s)}\int^{\infty}_0V^d_d(4\pi t)^{-d/2}e^{-m^2t}t^s\frac{\dt}{t}
\right).
\end{equation}
Now, all we need to do is to calculate (\ref{eq10}). Let us assume the case $d/2$ is not an integer. Using $x=m^2t$, we have
$$
\begin{aligned}
(\ref{eq10})&=\frac{\d}{\d s}\bigg|_{s=0}\left(
\frac{1}{\Gamma(s)}\int^{\infty}_0\frac{V^d_d}{(4\pi/m^2)^{d/2}}(m^2)^{-s}x^{-d/2}e^{-x}x^s\frac{\d x}{x}
\right)\\
&=\frac{\d}{\d s}\bigg|_{s=0}\left(
\frac{1}{\Gamma(s)}\frac{V^d_d}{(4\pi/m^2)^{d/2}}(m^2)^{-s}\Gamma(s-d/2)
\right)\\
&=\frac{V^d_d}{(4\pi/m^2)^{d/2}}\Gamma(-d/2).
\end{aligned}
$$
When $d/2$ is an integer, we can use a similar trick, however, we need be more careful as $\Gamma(s-d/2)$ diverges as $s\to0$. In this case, the computation goes as follows: 
$$
\begin{aligned}
(\ref{eq10})
&=\frac{\d}{\d s}\bigg|_{s=0}\left(
\frac{\Gamma(s-d/2)}{\Gamma(s)}\frac{V^d_d}{(4\pi/m^2)^{d/2}}(m^2)^{-s}
\right)\\
&=\frac{\d}{\d s}\bigg|_{s=0}\left(
\frac{1}{(s-d/2)(s-d/2+1)...(s-1)}\frac{V^d_d}{(4\pi)^{d/2}}(m^2)^{-s+d/2}
\right)\\
&=(-1)^{d/2}
\frac{2/d+2/(d-2)\cdots+1-\log(m^2)}{(d/2)!}\frac{V^d_d}{(4\pi)^{d/2}}m^d.
\end{aligned}
$$
\end{proof}
\end{proposition}

\begin{proposition}\label{04.3} The following limit holds:
$$
\begin{aligned}
&\lim_{u\to\infty}\int^{1}_0(\Theta_{DT_d(N)}^{\tilde{m}}(u^2t)-V^{d,N}_de^{-u^2\tilde{m}^2t}e^{-2du^2t}I_0(2u^2t)^d)\frac{\dt}{t}\\
&=\int^{1}_0\left\{\theta_{T_d(A,m)}(t)-V^d_de^{-m^2t}(4\pi t)^{-d/2}\right\}\frac{\dt}{t}.
\end{aligned}$$

\end{proposition}

Now, we state the key proposition of this section.
\begin{proposition} When $u\rightarrow\infty$,
$$
\begin{aligned}
&\mathcal{H}_{N(u)}(0)\\
&=\log \det(\Delta_{T_d(A)}+m^2)+\left\{
\begin{aligned}
&\frac{V^d_d}{(4\pi/m^2)^{d/2}}\Gamma(-d/2)~&d~odd\\
&(-1)^{d/2}
\frac{2/d+2/(d-2)\cdots+1-\log(m^2)}{(d/2)!}\frac{V^d_d}{(4\pi)^{d/2}}m^d~&d~even\\
\end{aligned}\right. \\ 
&+o(1).
\end{aligned}
$$
\end{proposition}

The following theorem is the main result of this section which gives the asymptotic expansion of log determinant of the discrete massive Laplacian on discrete tori.
\begin{theorem}
 As $u\to\infty$, the following holds:
\begin{equation}
 \log \det{\Delta^{\tilde{m}}_{DT_{d}(N)}}=V^{d,N}_d\mathcal{L}_{\tilde{m}}(0)+\mathcal{H}_{N(u)}(0),
\end{equation}
where $\mathcal{H}_{N(u)}(0)$ is as in the previous proposition and 
$$\mathcal{L}_{\tilde{m}}(0)=-\int^{\infty}_0(e^{-\tilde{m}^2t}e^{-2dt}I_0(2t)^d-e^{-t})\frac{\dt}{t}$$
\end{theorem}

\begin{example}\label{example1} When $d=1$ and $\tilde{m}=m/u$, using a result in \cite{louis2015asymptotics1}, we have
$$
\mathcal{L}_{\tilde{m}}(0)=\log{\left(\frac{2+\tilde{m}^2}{2}+\sqrt{\left(\frac{2+\tilde{m}^2}{2}\right)^2-1}\right)},
$$
Taylor expanding this function with respect to $\tilde{m}$ around $0$, we get
$$
\mathcal{L}_{\tilde{m}}(0)=\tilde{m}+o(\tilde{m}).
$$
Hence, in the one dimensional case, as $u\to\infty$,
\begin{align*}
  &\log\det {\Delta^{\tilde{m}}_{DT_{d}(N)}}
   \\&=am+\log \det(\Delta_{T_d(A)}+m^2)-am+o(1)\\
   &=\log \det(\Delta_{T_d(A)}+m^2)+o(1).
\end{align*}
\end{example}

Next, we will analyze $\mathcal{L}_{\tilde{m}}(0)$ when $d\geq2$. First, we need the following lemma from \cite{lang2013undergraduate} p. 340.
\begin{lemma}
Let $f$ be a function with two variables $t,x$ defined on $[0,\infty)\times[c,d]$. Assume that $f$ and $f_x:=\frac{\partial f}{\partial x}$ exists and continuous. Assume that 
$$
\int^\infty_{0} f_x(t,x)\dt
$$
converges uniformly for $x\in[c,d]$, and that
$$
g(x)=\int^\infty_{0} f(t,x)\dt
$$
converges for all x. Then g is differentiable, and 
$$
g'(x)=\int^\infty_{0} f_x(t,x)\dt
$$
\end{lemma}
When $d\geq2$, using the above lemma, we can show that $\mathcal{L}_{\tilde{m}}(0)$ is $d-1$ times differentiable as a function of $\tilde{m}$. This is the content of the following theorem. 

\begin{theorem}\label{lasttheorem} Let $\tilde{m}=m/u$ and $d\geq 2$. Then,
$$
\begin{aligned}
\mathcal{L}_{\tilde{m}}(0)=&-\int^{\infty}_0(e^{-2dt}I_0(2t)^d-e^{-t})\frac{\dt}{t}\\
&-\tilde{m}\int_0^{\infty}\dfrac{\partial f}{\partial\tilde{m}}(0,t)\,\dt-,\dots,\\
&-\tilde{m}^{d-1}\int_0^{\infty}\dfrac{\partial^{d-1} f}{\partial\tilde{m}^{d-1}}(0,t)\,\dt+o(\tilde{m}^{d-1}).
\end{aligned}
$$
where $f(t, \tilde{m})=(e^{-\tilde{m}^2t}e^{-2dt}I_0(2t)^d-e^{-t})/{t}$. Moreover, the terms with odd order derivative approach to zero as $\tilde{m}\to 0$.
\end{theorem}

Before proving the theorem, let us prove a lemma.
\begin{lemma}\label{keylastlemma} The improper integral
$$
\int^{\infty}_0\tilde{m}^{a-b}e^{-\tilde{m}^2t}e^{-2dt}I_0(2t)^dt^{a-1}\dt
$$
converges uniformly for all $a,b\in \mathbb{R}$ such that $a\geq b \geq 0$ and $a+b<d$.
\end{lemma}
\begin{proof}
We need to show that for any $\epsilon>0$, there exists a $A>0$ such that for all $\tilde{m}\in[0,1]$,
\begin{eqnarray}
\Big|\int^{\infty}_A\tilde{m}^{a-b}e^{-\tilde{m}^2t}e^{-2dt}I_0(2t)^dt^{a-1}\dt\Big|<\epsilon.
\end{eqnarray}
Using Lemma 4.1 from \cite{chinta2010zeta}, given $A>0$ there exists $C>0$ such that 

\begin{equation}\label{inequality4.5}
 \Big|\int^{\infty}_A\tilde{m}^{a-b}e^{-\tilde{m}^2t}e^{-2dt}I_0(2t)^dt^{a-1}\dt\Big|
\leq C\Big|\int^{\infty}_A\tilde{m}^{a-b}e^{-\tilde{m}^2t}t^{-d/2}t^{a-1}\dt\Big|
\end{equation}

Let $c=A\tilde{m}^2$. Now, from (\ref{inequality4.5}) we have:
$$
\begin{aligned}
C\Big|\int^{\infty}_A\tilde{m}^{a-b}e^{-\tilde{m}^2t}t^{a-d/2-1}\dt\Big|=&C\Big|\int^{\infty}_c\tilde{m}^{d-a-b}e^{-t}t^{a-d/2-1}\dt\Big|\\
\leq&C\Big|\int^{\infty}_c(\sqrt{\frac{c}{A}})^{d-a-b}e^{-t}t^{a-d/2-1}\dt\Big|\\
\leq&C\Big|\int^{\infty}_c\sqrt{\frac{c}{A}}e^{-t}t^{a-d/2-1}\dt\Big|.
\end{aligned}
$$
We will show that $C\Big|\int^{\infty}_c\sqrt{\frac{c}{A}}e^{-t}t^{a-d/2-1}\dt\Big|$ is bounded. For this we will analyze the the cases $c\geq 1$ and $c<1$. First, consider the $c\geq1$ case. Let $x=\log t$ and $r=\log c$. Then, $e^x>r/2+x^2$ for all $x>c$. Hence,
$$
\begin{aligned}
C\Big|\int^{\infty}_c\sqrt{\frac{c}{A}}e^{-t}t^{a-d/2-1}\dt\Big|\leq&C\Big| e^{r/2}\int^{\infty}_r\sqrt{\frac{1}{A}}e^{-r/2-x^2}e^{x(a-d/2)}\d x\Big|\\
\leq&C\Big|\int^{\infty}_0\sqrt{\frac{1}{A}}e^{-x^2}e^{x(a-d/2)}\d x\Big|\\
\leq&\frac{C'}{\sqrt{A}}.
\end{aligned}
$$

When $c<1$ and $a-d/2>0$ or $d$ odd,
$$
\begin{aligned}
C\Big|\int^{\infty}_c\sqrt{\frac{c}{A}}e^{-t}t^{a-d/2-1}\dt\Big|\leq&C\Big|\int^{\infty}_0\sqrt{\frac{c}{A}}e^{-t}t^{a-d/2-1}\dt\Big|\\
\leq&C\Big|\int^{\infty}_0\sqrt{\frac{1}{A}}e^{-t}t^{a-d/2-1}\dt\Big|\leq\frac{C'}{\sqrt{A}}.
\end{aligned}
$$

When $c<1$ and $a-d/2\leq 0$,
$$
\begin{aligned}
C\Big|\int^{\infty}_c\sqrt{\frac{c}{A}}e^{-t}t^{a-d/2-1}\dt\Big|\leq&C\Big|\int^{\infty}_1\sqrt{\frac{c}{A}}e^{-t}t^{a-d/2-1}\dt\Big|+C\Big|\int^{1}_c\sqrt{\frac{c}{A}}e^{-t}t^{a-d/2-1}\dt\Big|\\
\leq&\frac{C}{\sqrt{A}}+C\Big|\int^{1}_c\sqrt{\frac{1}{A}}t^{-1}\dt\Big|\leq-C\sqrt{\frac{c}{A}}\log c+\frac{C}{\sqrt{A}} \leq\frac{C'}{\sqrt{A}}.
\end{aligned}
$$
Now, (4.5) follows from choosing $A>0$ so that $\frac{C'}{\sqrt{A}}<\epsilon.$
\end{proof}
%\begin{pf4.7}
Now, we prove Theorem \ref{lasttheorem}. We need to show that $\mathcal{L}_{\tilde{m}}(0)$ is $d-1$ times differentiable with respect to $\tilde{m}$, then the theorem will follow. We first observe that $f(t,\tilde{m})$ is continuous on $(0,\infty)\times [0,r]$ for any $r>0$.

If we assume $\mathcal{L}_{\tilde{m}}(0)$ is sufficiently differentiable, then any $k$th derivative is going to be sum of the terms of the form $\int \tilde{m}^{a-b}e^{-\tilde{m}^2t-2dt}I_0(2t)^dt^{a-1}dt$ for $a+b=k$. This means that if we want to show $\mathcal{L}_{\tilde{m}}(0)$ has $k$th order derivative inductively, we only need to check such integrals exist, which is the content of the Lemma \ref{keylastlemma}. This completes proof of Theorem \ref{lasttheorem}.
%\end{pf4.7}

 Hence, the asymptotic expansion of $V_{d}^{d,N(u)}\mathcal{L}_{{m}/{u}}(0)$ as $u\to\infty$, from Theorem \ref{lasttheorem}, has the following form:
\begin{eqnarray}
V_{d}^{d,N(u)}\mathcal{L}_{{m}/{u}}(0)=u^d\mathrm{Vol}(T_d(A))\mathcal{L}_{{m/u}}(0)=\sum_{k=0}^{d-1}m^ku^{d-k}c_k+o(u).
\end{eqnarray}

\begin{remark}
If we further assume that $$ V_{d}^{d,N(u)}\mathcal{L}_{{m}/{u}}(0)=\sum_{k=0}^{d-1}m^ku^{d-k}c_k+\widetilde{C}\log u+C+o(1)$$ as $u\to\infty$ where $\widetilde{C}$ and $C$ are constants, then the constant term in the asymptotic expansion of $\log\det(\Delta_{DT_d(N(u)}^{m/u})$ is given by $\log\det(\Delta_{T_d(A)}^{m})+C+C'$ where $C'$ is such that $\mathcal{H}_{N(u)}(0)=\log\det(\Delta_{T_d(A)}^{m})+C'.$ This means that, in this case, $\log\det(\Delta_{T_d(A)}^{m})$ can be read off from the asymptotic expansion of $\log\det(\Delta_{DT_d(N(u)}^{m/u})$ as $u\to\infty$. We have seen in example \ref{example1} that for $d=1$ our assumption holds. For $d=2$, it is shown in \cite{duplanticr1988exact} (p.361) that the constant term in the asymptotic expansion of $\log\det(\Delta_{DT_2(N(u)}^{m/u})$ is given by
$$
\log\det(\Delta_{T_2(A)}^{m})+\frac{4G}{\pi}V(A)+m^2V(A)\frac{5}{4\pi}\log{2}. 
$$

%Notice that in $\mathcal{L}_m(0)$, there is a $\tilde{m}^2$, which is of order $1/u^2$. Thus we expand $\mathcal{L}_m(0)$ with respect to u or say $\tilde{m}$ to certain order as shown. There this term will also contribute some term of the same order of $\mathcal{H}_{N(u)}$, but our method can not give the result of this order. Naive guess would be this might contribute some term of the form $\log(\tilde{m}^2)$
\end{remark}

\subsection{Two dimension spanning forest}

Consider two dimensional massive Laplacian on the interval and the torus. The eigenvalues of this two is already mentioned in former section. 
Then we have for the discrete hypercube $N'=(n_1,n_2)$, the determinant is:
$$
\det\Delta_{L_2(N')}^m)^4=\prod\left(4+m^2-2\cos(\frac{m_1\pi}{n_1})-2\cos(\frac{m_2\pi}{n_2})\right)
$$
here $m_i$ runs from 1 to $n_i-1$, and $m_i$ can not all be zero.
And for the torus of size $N=(2n_1,2n_2)$, we have:
$$
\det\Delta_{DT_2(N)}^m=\prod_{m_1,m_2}\left(4+m^2-2\cos\left(\frac{m_1\pi}{n_1}\right)-2\cos\left(\frac{m_2\pi}{n_2}\right)\right)
$$
here $m_i$ runs from 0 to $2n_i-1$, and $m_i$ can not all be zero.

We could reach that for the interval with Dirichlet boundary condition:
$$
\det\Delta_{L_2(N')}^m)^4=\prod_{m_1=1}^{m_1=n_1-1}\prod_{m_2=1}^{m_2=n_2-1}\left(4+m^2-2\cos\left(\frac{m_1\pi}{n_1}\right)-2\cos\left(\frac{m_2\pi}{n_2}\right)\right)
$$

For the torus case, we could rewrite it into the form:
$$\begin{aligned}
\det\Delta_{DT_2(N)}^m=
&\prod_{m_1=1}^{m_1=n_1-1}\left(2+m^2-2\cos\left(\frac{m_1\pi}{n_1}\right)\right)^2\prod_{m_2=1}^{m_2=n_2-1}\left(2+m^2-2\cos\left(\frac{m_2\pi}{n_2}\right)\right)^2\\
&\prod_{m_1=1}^{m_1=n_1-1}\left(6+m^2-2\cos\left(\frac{m_1\pi}{n_1}\right)\right)^2\prod_{m_2=1}^{m_2=n_2-1}\left(6+m^2-2\cos\left(\frac{m_2\pi}{n_2}\right)\right)^2\\
&\prod_{m_1=1}^{m_1=n_1-1}\prod_{m_2=1}^{m_2=n_2-1}\left(4+m^2-2\cos\left(\frac{m_1\pi}{n_1}\right)-2\cos\left(\frac{m_2\pi}{n_2}\right)\right)^4(8+m^2)(4+m^2)^2
\end{aligned}
$$

Compare this two formulas, we have:
$$\begin{aligned}
&\frac{\det\Delta_{DT_2(N)}^m}{(\det\Delta_{L_2(N')}^m)^4}=(8+m^2)(4+m^2)^2\times\\
&\prod_{m_1=1}^{m_1=n_1-1}\left(6+m^2-2\cos\left(\frac{m_1\pi}{n_1}\right)\right)^2\prod_{m_2=1}^{m_2=n_2-1}\left(6+m^2-2\cos\left(\frac{m_2\pi}{n_2}\right)\right)^2\times\\
&\prod_{m_1=1}^{m_1=n_1-1}\left(2+m^2-2\cos\left(\frac{m_1\pi}{n_1}\right)\right)^2\prod_{m_2=1}^{m_2=n_2-1}\left(2+m^2-2\cos\left(\frac{m_2\pi}{n_2}\right)\right)^2
\end{aligned}
$$
and we can use the following formula in \cite{louis2015formula} to simplify it:
$$
\prod_{k=1}^{n-1}\left(2x-2\cos\left(\frac{k\pi}{n}\right)\right)=(x+\sqrt{x^2-1})^n+(x-\sqrt{x^2-1})^n-2.
$$

\begin{remark}
Using the same idea, we can give a relationship between the the determinant of massive Laplacian on torus and the hypercube with free boundary condition which generalizes a key result of \cite{louis2017asymptotics} in the sense that if we take $m=0$ our result matches with the result in \cite{louis2017asymptotics}. 
\end{remark}

\section{Acknowledgements}
We want to thank Alberto Cattaneo for many inspiring discussions. We also want to thank an anonymous referee for the suggestions for the improvement.

%%%%%%%%%%%%%%%%%%%%%%%%%%%%%%%%%%%%%%%%%%

\bibliographystyle{amsplain}
\bibliography{reference}
Yuhang Hou,
Mathematisches Institute, 
University of Freiburg, 
Freiburg im Breisgau, Germany\\
email: yuhang.hou@math.uni-freiburg.de
\\
Santosh Kandel, 
Institut f\"ur Mathematik,
University of Zurich, Zurich, Switzerland\\
email: skandel1@alumni.nd.edu

\end{document}